\documentclass[11pt]{article}

\usepackage{booktabs} % For formal tables
\usepackage{natbib}
\usepackage{amssymb,amsmath,amstext}
\usepackage{hyperref}
\usepackage{fullpage}
\usepackage{xspace}
\usepackage{latexsym}
\usepackage{verbatim}
\usepackage{multirow}
\usepackage{algorithm, algorithmic}
\usepackage{enumitem}

\newtheorem{theorem}{Theorem}[section]

\newtheorem{lemma}[theorem]{Lemma}

\newtheorem{corollary}[theorem]{Corollary}

\newtheorem{remark}{Remark}[section]

\newcommand{\qed}{\mbox{\ \ \ }\rule{6pt}{7pt} \bigskip}

\newenvironment{proof}{\noindent{\em Proof:}}{\hfill\qed}

\makeatletter
\newenvironment{oneshot}[1]{\noindent {\bf{#1}}. }
\makeatother
\newcommand{\expect}{\mathop{\operatorname{\bf E}}}
\newcommand{\ex}{\expect}

\newcommand{\pr}{\mathop{\operatorname{\bf Pr}}}
\renewcommand{\Pr}{\pr}
\newcommand\prob{\mbox{\bf Pr}}

\newcommand{\Real}{\mbox{\rm\bf R}}

\newcommand{\nikhil}[1]{}

\newcommand{\opt}{\mbox{\rm OPT}}

\newcommand{\asi}{adversarial stochastic input}
\newcommand{\aijk}{a_{ijk}}
\newcommand{\bij}{ b_{ij}}
\newcommand{\vij}{ v_{ij}}
\newcommand{\xij}{ x_{ij}}

\newcommand{\xjk}{x_{j,k}} 
\newcommand{\xjkt}{x_{j,k,t}} 
\newcommand{\wjk}{w_{j,k}}
\newcommand{\wijk}{w_{ijk}}
\newcommand{\rand}{R}

\newcommand{\optr}{W_{\rand}}
\newcommand{\opte}{W_{E}}
\newcommand{\optei}[1][i]{W_{E,#1}}

\newcommand{\ho}{P}
\newcommand{\hoe}{\widetilde{\ho}}

\newcommand{\xin}{X_i}
\newcommand{\yin}{Y_i}
\newcommand{\xina}{X_i^A}
\newcommand{\yina}{Y_i^A}

\newcommand{\xit}[1][t]{X_{i,#1}}
\newcommand{\yit}[1][t]{Y_{i,#1}}

\newcommand{\xito}[1][t]{X_{i,#1}^{\ast}}
\newcommand{\yito}[1][t]{Y_{i,#1}^{\ast}}

\newcommand{\xjko}{x_{jk}^{\ast}}

\newcommand{\xitoe}[1][t]{\widetilde{X}_{i,#1}}
\newcommand{\yitoe}[1][t]{\widetilde{Y}_{i,#1}}

\newcommand{\xita}[1][t]{X_{i,#1}^A}
\newcommand{\yita}[1][t]{Y_{i,#1}^A}

\newcommand{\sxir}[1][s]{S_{#1}^{\stage}\left(\xin\right)}
\newcommand{\syir}[1][s]{S_{#1}^{\stage}\left(\yin\right)}

\newcommand{\sxia}[1][s]{S_{#1}\left(\xina\right)}
\newcommand{\syia}[1][s]{S_{#1}\left(\yina\right)}
\newcommand{\sxiar}[1][s]{S_{#1}^{\stage}\left(\xina\right)}
\newcommand{\syiar}[1][s]{S_{#1}^{\stage}\left(\yina\right)}

\newcommand{\phixi}[1][s]{\phi_{i,#1}}
\newcommand{\phiyi}[1][s]{\psi_{i,#1}}
\newcommand{\phixir}[1][s]{\phi_{i,#1}^{\stage}}
\newcommand{\phiyir}[1][s]{\psi_{i,#1}^{\stage}}

\newcommand{\epsx}[1][\stage]{\epsilon_{x,#1}}
\newcommand{\epsy}[1][\stage]{\epsilon_{y,#1}}

\newcommand{\uf}{\mathcal{F}}
\newcommand{\ufr}{\mathcal{F}_{\stage}}
\newcommand{\stage}{r}

\newcommand{\tr}[1][\stage]{t_{#1}}
\newcommand{\er}[1][\stage]{e_{#1}}
\newcommand{\zr}[1][\stage]{Z_{#1}}
\newcommand{\ir}[1][\stage]{\mathcal{I}_{#1}}

\newcommand{\wm}{w_{\max}}

\newcommand{\no}{\mathrm{NO}}
\newcommand{\dmi}{\rho_i}

\newcommand{\reso}{\mathcal{A}}
\newcommand{\req}{\mathcal{J}}
\newcommand{\optn}{\mathcal{K}}

\begin{document}
	
	\title{Near Optimal Online Algorithms and Fast Approximation Algorithms for 
		Resource Allocation Problems\thanks{
			Part of this work was done when the second author was a researcher at 
			Microsoft Research, Redmond, and the third, fourth authors were interns at 
			Microsoft Research, Redmond.
	}}
	\author{
		Nikhil R. Devanur\thanks{Microsoft Research. 
			\tt{nikdev@microsoft.com}.}
		\and Kamal Jain\thanks{Faira. \tt{kamaljain@gmail.com}.}
		\and Balasubramanian Sivan\thanks{Google Research. 
			\tt{balusivan@google.com}.}
		\and Christopher A. Wilkens \thanks{Facebook Research. \tt{cwilkens@berkeley.edu}.}
	}
	
	\date{}
	\maketitle{}
	\thispagestyle{empty}

\begin{abstract}
We present prior robust algorithms for a large class of resource allocation problems where requests arrive one-by-one (online), drawn independently from an {\em unknown} distribution at every step. We design a single algorithm that, for every possible underlying distribution, obtains a $1-\epsilon$ fraction of the profit obtained by an algorithm that knows the entire request sequence ahead of time. The factor $\epsilon$ approaches $0$ when no single request consumes/contributes a significant fraction of the global consumption/contribution by all requests together. We show that the tradeoff we obtain here that determines how fast $\epsilon$ approaches $0$, is near optimal: we give a nearly matching lower bound showing that the tradeoff cannot be improved much beyond what we obtain. 

Going beyond the model of a static underlying distribution, we introduce the {\em adversarial stochastic input} model, where an adversary, possibly in an adaptive manner, controls the distributions from which the requests are drawn at each step. Placing no restriction on the adversary, we design an algorithm that obtains a $1-\epsilon$ fraction of the optimal profit obtainable w.r.t. the worst distribution in the adversarial sequence. Further, if the algorithm is given one number per distribution, namely, the optimal profit possible for each of the adversary's distribution, we design an algorithm that achieves a $1-\epsilon$ fraction of the  weighted average of the optimal profit of each distribution the adversary picks. 

In the offline setting we give a fast algorithm to solve very large LPs with both packing and covering constraints. We give algorithms to approximately  solve (within a factor of $1+\epsilon$) the mixed packing-covering problem with $O(\frac{\gamma m \log (n/\delta)}{\epsilon^2})$ oracle calls where the constraint matrix of this LP has dimension $n\times m$, the success probability of the algorithm is $1-\delta$, and $\gamma$ quantifies how significant a single request is when compared to the sum total of all requests. 

We discuss implications of our results to several special cases including online combinatorial auctions, network routing and the adwords problem.

\end{abstract}

\section{Introduction \& Summary of Results}\label{online:sec:intro}
%\label{sec:intro}

\newcommand{\pj}{\mathcal{P}_j}
\newcommand{\xj}{\mathbf{x}_j}
\newcommand{\aij}{\mathbf{a}_{i,j}}
\newcommand{\bbij}{\mathbf{b}_{i,j}}
\newcommand{\wj}{\mathbf{w}_j}
\newcommand{\wij}{\mathbf{w}_{i,j}}
\newcommand{\alphab}{\mathbf{\alpha}}
\newcommand{\bigo}[1]{O\left( #1 \right)}
\newcommand{\bigomega}[1]{\Omega\left( #1 \right)}

There has been an increasing interest in online algorithms for resource 
allocation problems motivated by their wide variety of applications in Internet 
advertising, allocating multi-leg flight seats for customers online, allocating 
goods to customers arriving online in a combinatorial auction etc. Designing 
efficient resource allocation algorithms has significant scientific and 
commercial value. The traditional computer science approach to deal with 
uncertain future inputs has been the worst-case competitive analysis. Here 
nothing is assumed about the sequence of requests that arrive online, and the 
benchmark is the optimal algorithm that knows the entire sequence of requests 
ahead of time. Several problems in this space have been analyzed in the 
traditional framework, exemplified, for instance, in the well-studied {\em 
Adwords} problem introduced by~\citet{MSVV05}. While worst-case 
analysis is a robust framework, for many problems it leads to pessimistic 
bounds that rule out obtaining more than a constant fraction of the optimal 
profit. Consequently, there has been a drive in the last few years to go beyond 
worst-case analysis.  A frequently used alternative is to perform stochastic 
analysis: assume that the input is drawn from a {\em known} distribution and 
optimize the objective w.r.t. this distribution. While stochastic analysis 
circumvents the impossibility results in worst-case analysis, any error in the 
knowledge of distribution could render the algorithm suboptimal, and sometimes 
even infeasible. 

%where an algorithm assigns keywords arriving online to bidders to maximize profit without violating bidders' budget constraints, is a special case of the resource allocation framework. Several problems in this space have been analyzed in the traditional framework for online algorithms, namely, worst-case competitive analysis. Here nothing is assumed about the sequence of requests that arrive online, and the benchmark is the optimal algorithm that knows the entire sequence of requests ahead of time. For many problems worst-case competitive analysis is not entirely satisfactory: it leads to pessimistic bounds that rule out obtaining more than a constant fraction of the optimal profit.  Consequently, there has been a drive in the last few years to go beyond worst-case analysis.  A frequently used alternative is to perform stochastic analysis: assume that the input is drawn from a {\em known} distribution and optimize the objective w.r.t. this distribution. While stochastic analysis circumvents the impossibility results in worst-case analysis, any error in the knowledge of distribution could render the algorithm suboptimal, and sometimes even infeasible. 

In this paper, we study a middle ground between worst-case and stochastic 
settings. We assume that the input is drawn from an underlying distribution 
that is {\em unknown} to the algorithm designer. We present a single algorithm, 
that for every distribution performs nearly as well as the optimal algorithm 
that knows the entire sequence of requests ahead of time. In this sense, the 
algorithm is {\em prior robust}. 

We now give an informal description of the resource allocation framework and our main contributions. See Section~\ref{online:sec:prelim} for a formal description and theorem statements.  We consider a resource allocation setting where requests arrive online; every request can be served by some subset of several available options; each (request, option) pair consumes some amount of every resource, and generates some profit. There is a given budget for each resource. 
Requests are drawn i.i.d. from an {\em unknown} distribution. 
The goal is to maximize the total profit generated while making sure that the total consumption of each resource is no more than the corresponding budget. 
%a given objective function (this paper handles optimizing the sum and minimum of linear functions). 
We compare the profit of the algorithms against the offline optimum and prove competitive ratio bounds. 
Even for very restricted special cases of this problem, the worst-case setting  
cannot yield anything beyond a $1-\frac{1}{e}$ competitive 
ratio~\cite{KP96,MSVV05}.
While the stochastic setting with a fully known distribution can give near 
optimal performance guarantees, 
it often leads to very distribution dependent algorithms (e.g. see~\cite{AHL12} for the special case of the adwords problem, which requires knowledge of the entire distribution to perform the optimization). 
Hence both these approaches are not satisfactory, and this problem lends itself well to the middle ground of prior robust analysis. 

Going beyond i.i.d., our work introduces the {\em adversarial stochastic input} (ASI) model as a more realistic
model for analyzing online algorithmic problems. Here the distribution from which the requests arrive is allowed to change over time (unlike i.i.d., where it stays identical for every request). The adversary decides how to pick the distributions, and is even allowed to pick them adaptively. 
%Our algorithms guarantee a profit close to the optimal profit for the worst of these distributions. 
For many practical applications such as in display advertising, the distribution of requests shows trends that change over the course of time: mornings are different from evenings and weekdays are different from weekends. 
Thus a time varying distributional model is more realistic than the i.i.d. model.  
A keen reader might notice that the above description includes the worst-case setting as well, therefore we have to make some extra assumptions, either by restricting how these distributions can be picked, or by allowing the algorithm some extra information about the distributions. We will describe these in greater detail later on.

Apart from the theoretical contribution, {\em the algorithms we design for the ASI models were used to completely overhaul the display advertising management system at Microsoft}, leading to a significant improvement in revenue ($\approx 10 \%$), better system manageability and enabling new capabilities.\footnote{This system had been globally operational from 2011  to  2015, when Microsoft made a deal with AOL to allow AOL to sell the display advertisement on behalf of Microsoft. }
We believe that our results  make a significant contribution to the search for ``allows-positive-results-yet-realistic'' models for online algorithms.
%A closely related, and slightly more general model that we refer to often is the {\em random permutation} model (introduced by Goel and Mehta~\cite{GM08}), that assumes that the adversary picks the set of keywords, but the order in which the keywords arrive is chosen uniformly at random.  

%We summarize our results in more detail below. The formal definitions of the problems are presented in Section~\ref{online:sec:prelim}. 

\paragraph*{First Result: Near-Optimal Prior Robust Online Algorithms for Resource Allocation
Problems}
A key parameter on which algorithms for several resource allocation problems 
depend on is the relative significance of any single request when compared to 
the entire sequence of requests. For instance, for the special case of the 
Adwords problem, this is the ratio of a single bid to an advertiser's budget. 
For the Adwords problem,~\cite{MSVV05} and~\cite{BJN07} design an algorithm 
that achieves a worst
case competitive ratio that tends to $1-1/e$ as the bid to budget ratio (which
we denote by $\gamma$) tends to 0.\footnote{Note that $\gamma$ approaching zero is the
easiest case. Even with $\gamma$ approaching zero, $1-1/e$  is the best competitive ratio
that any randomized algorithm can achieve in the worst case, illustrating how 
worst-case analysis leads to pessimistic bounds.}~\cite{DH09} studied the same 
problem in the random permutation model, and showed that the competitive
ratio tends to 1 as $\gamma$ tends to 0.  This result showed that competitive
ratio of algorithms  in stochastic models  could be  much better than that of
algorithms in the worst case.  The important question since then has been to
determine the optimal trade-off between $\gamma$ and the competitive ratio.
\cite{DH09} showed how to get a 1- $O(\epsilon)$ competitive ratio when 
$\gamma$
is  at most $O(\frac{\epsilon^3}{n\log (mn/\epsilon) })$ where $n$ is the number
of advertisers and $m$ is the number of keywords.  Subsequently~\cite{AWY09} 
improved the bound on $\gamma$ to
$O(\frac{\epsilon^2}{n\log (m/\epsilon) })$.  The papers of~\citet{FHKMS10} and 
~\citet{AWY09} have also shown that
the technique of~\cite{DH09} can be extended to other online problems.

{\em The first main result in this paper is the following $3$-fold improvement of
previous results (Theorems~\ref{online:thm:online_ra} and~\ref{online:thm:lb}), for the i.i.d. with unknown distributions model}. All our results apply to the general class of problems that we call the {\em resource allocation framework}.
A formal definition of the framework is presented in Section~\ref{online:subsec:online_ra} and a discussion of many interesting special  cases including online network routing and online combinatorial auctions is presented in Section~\ref{online:sec:special_cases}.
\begin{enumerate}[leftmargin=*]
\item We give an algorithm which guarantees a $1-\epsilon$ approximation factor when
$\gamma = O(\frac{\epsilon^2}{\log(n/\epsilon)})$. 
\item We show that our bound on $\gamma$ is {\em almost optimal}; we show that no algorithm, even if it knew the distribution, can guarantee a $1-\epsilon$ approximation when $\gamma = \omega(\frac{\epsilon^2}{\log (n) })$. 
\item Our algorithms lend themselves to natural generalizations that provide identical guarantees in the more general {\em \asi} (ASI) model that was described earlier. We provide three different versions of the ASI model in Section~\ref{online:sec:asiGWL}.  
\end{enumerate}

\paragraph*{Significance} 
\begin{enumerate}[leftmargin=*]
\item Regarding the bound on $\gamma$, we remove a factor of $n$ from $\gamma$, making the algorithm more practical. 
Consider for instance the Adwords problem and suppose that the bids
are all in [0,1]. The earlier bound implies that the advertiser budgets need to be of the
order of $n \log n/\epsilon^2$ in order to get a $1-\epsilon$ competitive algorithm,
where $n$ is the number of advertisers.  With realistic values for these
parameters, it seems unlikely that this condition would be met.  While with the
improved bounds presented in this paper, we only need the advertiser budget to be of the
order of $\log n /\epsilon^2$ and this condition is met for reasonable values of
the parameters. Furthermore, in the more general resource allocation framework,
the previous best upper bound on $\gamma$ is from Agrawal, Wang and
Ye~\cite{AWY09} and equals $O(\frac{\epsilon^2}{n\log (mK/\epsilon) })$.  Here
$K$ is the number of available ``options'' (see Section~\ref{online:subsec:online_ra})
and in typical applications like network routing, $K$ could be exponential in
$n$, and thus, the factor saved by our algorithm becomes quadratic in $n$. 

\item Our ASI models are realistic models of time varying distributions for which we present algorithms with asymptotically optimal performance guarantees. 
We consider three different benchmarks, each progressively stronger than the previous, and require different levels of information about the distributions to achieve near optimal performance guarantees. For the weakest benchmark, we need just one parameter from the distribution, 
while for the strongest benchmark, we still need only $2mn$ parameters. Note 
that the distributions themselves can have an arbitrarily large support 
size\footnote{ We even allow continuous distributions which have an infinite 
support.} and hence the amount of information we need is much smaller than the 
description of all the distributions. Our results for the ASI model can be 
thought of as generalizations of the ``Prophet Inequality".\footnote{The 
Prophet Inequality is essentially a $1/2$-competitive algorithm for the 
following problem: a sequence of values is drawn independently from different 
distributions and presented one at a time. The algorithm may pick at most one 
of these values in an online manner, given some knowledge of the distributions, 
such as the expectation of the maximum of these values. The goal is to maximize 
the value picked. See~\cite{Cahn84}} Finally, as mentioned earlier, our 
algorithms 
for this model have made a significant impact on the practice of display 
advertising management at Microsoft. 

%\item In the \asi\ model, we show that if the algorithm is provided with one parameter for each distribution the adversary picks, it can get a $1-\epsilon$ fraction of the weighted average of optimal profit of each distribution the adversary picks. This level of generalization and robustness fits well with the display ads business requirement in Microsoft (see Section~\ref{online:subsec:online_ra} for a quick definition of the display ads problem, which is a special case of the resource allocation problem). Distributions determining users' webpage requests necessarily change over time of the day, and hence a static model like the i.i.d. model or the random permutation model is not suitable here. The ASI model captures time varying distributions and requires that the algorithm designer know just one parameter per distribution. This is exactly the sense in which Microsoft uses our algorithm: these parameters are estimated by Microsoft based on past data, and then our algorithm is run. Our algorithm has been operational globally since the summer of 2011. 

%\item Special cases of this problem include the adwords problem, display ads problem, online network routing and online combinatorial auctions. Our algorithms when specialized for online combinatorial auctions
\end{enumerate}

\paragraph*{Second Result: Prior Robust $1-1/e$ approximation Greedy Algorithm
for Adwords}
%%%%%%%%% Motivate this via greedy %%%%%%%%%%%%%%%%%%%%%%%%%%%%%%%%%%%%%%%%%%%%%%%%%
A natural algorithm for the Adwords problem that is widely used for its
simplicity is the greedy algorithm: always match an incoming query to the
advertiser that has the maximum effective bid (the minimum of bid and remaining
budget) for that query. Because of its wide use, previously the performance of
the greedy algorithm has been analyzed by Goel and Mehta~\cite{GM08} who showed
that in the random permutation and the i.i.d.  models, it has a competitive
ratio of $1-1/e$ with an assumption which is essentially that $\gamma$ tends to
0. 

It has been an important open problem to analyze the performance of greedy
algorithm in a stochastic setting for unbounded $\gamma$, i.e., for all
$0\leq\gamma\leq 1$.  The best factor known so far is $1/2$, and this works for
the worst case also. 
%(For the special case of online
%bipartite matching, in the case of i.i.d input with a {\em known} distribution,
%recent series of results achieve a ratio of better than 1-1/e, for instance by
%Feldman et al.~\cite{FMMM09} and  Bahmani and Kapralov~\cite{BK10}. The best
%ratio so far is .702 by Manshadi, Gharan and Saberi~\cite{SS10}. The same online
%bipartite matching has been recently studied in the random permutation model by
%Mahdian and Yan~\cite{MY11} and by Karande, Mehta and Tripathi~\cite{KMT11}.
%The best ratio so far is $0.696$ by Mahdian and Yan~\cite{MY11}.)
Nothing better was known, even in the stochastic models.   {\em The second result
in this paper is that for the Adwords problem in the i.i.d.  unknown
distributions model, with no assumption on $\gamma$ (i.e., $\gamma$ could be as
big as $1$), the greedy algorithm gets an approximation factor of $1-1/e$ against
the optimal fractional solution to the expected instance (Theorem
\ref{online:thm:adwords_greedy})}.  

Our proof technique for this result has been subsequently used to prove a similar result for the greedy algorithm in online submodular welfare maximization~\cite{KPV13}. We note here that there are other algorithms that achieve a $1-1/e$ approximation for the Adwords problem with unbounded $\gamma$, but the greedy algorithm is the only prior robust (i.e., distribution independent) algorithm known, and it is quite simple too. For example Alaei, Hajiaghayi and Liaghat ~\cite{AHL12} design a randomized algorithm that obtains a $1-1/e$ approximation, but requires the knowledge of the entire distribution. Devanur, Sivan and Azar~\cite{DSA12} design a deterministic algorithm that obtains a $1-1/e$ approximation, but requires a few parameters from the distribution.

\paragraph*{Third Result: Fast Approximation Algorithms for Mixed Packing and
Covering Integer Programs}
\cite{CCDJS10} considered the following (offline) problem: given
a lopsided bipartite graph $G=(L,R,E)$, that is a bipartite graph where $m = |L|
\gg |R| =n$, does there exist an assignment $M: L \rightarrow R$ with $(j,M(j))
\in E $ for all $j \in L$, and such that for every vertex $i \in R$,
$|M^{-1}(i)| \geq B_i$ for some given values $B_i$.  Even though this is a
classic problem in combinatorial optimization with well known polynomial time
algorithms, the instances of interest are too large to use traditional
approaches to solve this problem. (The value of $m$ in particular is very
large.) The approach used by~\citet{CCDJS10} was to essentially design an online
algorithm in the i.i.d. model: choose vertices from $L$ uniformly at random and
assign them to vertices in $R$ in an online fashion.  The online algorithm is
guaranteed to be close to optimal, as long as sufficiently many samples are
drawn.  Therefore it can be used to solve the original problem (approximately):
the online algorithm gets an almost satisfying assignment if and only if the
original graph has a satisfying assignment (with high probability). 

{\em The third result in this paper is a generalization of this result to get
fast approximation algorithms for a wide class of mixed packing and covering
integer programs (IPs)
inspired by problems in the resource allocation framework (Theorem
\ref{online:thm:coveringPacking})}.  Problems in the resource allocation
framework where the instances are too large to use traditional algorithms occur
fairly often, in particular in the management of display advertising systems, where these algorithms are being used. 
Formal statements and a more detailed discussion are presented  in Section~\ref{online:subsec:intro_covering_packing}.

\paragraph*{High Level Description of Techniques} The underlying idea used for
all these results can be summarized at a high level as thus: consider a
hypothetical algorithm called {\em Hypothetical-Oblivious} that knows the
distribution from which the input is drawn and uses an optimal solution w.r.t.
this distribution.  Now suppose that we can analyze the performance of
Hypothetical-Oblivious by considering a potential function and showing that it
decreases by a certain amount in each step. 
%(A caveat  is that this has to be somewhat independent of what the algorithm
%has done so far.)
Now we can design an algorithm that does not know the distribution as follows:
consider the same potential function, and in every step choose the option that
minimizes the potential function.  Since the algorithm minimizes the potential
in each step, the decrease in the potential for this algorithm is better than
that for Hypothetical-Oblivious and hence we obtain the same guarantee as that
for Hypothetical-Oblivious. The choice of potential function varies across the
results; also, whether we minimize or maximize the the potential function
varies across the results.  

For instance, our first result (Theorem~\ref{online:thm:online_ra}), the
performance of Hypothetical-Oblivious is analyzed using Chernoff bounds.  The
Chernoff bounds are proven by showing bounds on the expectation of the moment
generating function of a random variable.  Thus the potential function is the
sum of the moment generating functions for all the random variables that we
apply the Chernoff bounds to. The proof shows that in each step this potential
function decreases by some multiplicative factor. The algorithm is then designed
to achieve the same decrease in the potential function.  A  particularly
pleasing aspect about this technique is that we obtain very simple proofs. 
E.g., the proof of the second result mentioned above (that greedy is $1-1/e$ 
competitive, Theorem \ref{online:thm:adwords_greedy} ) is extremely
simple: the potential function in this case is simply the total amount of unused
budgets and we show that this amount (in expectation) decreases by a factor of
$1-1/m$ in each step where there are $m$ steps in all. 

\paragraph*{Multiplicative-Weight Updates} Our techniques and the resulting
algorithms for our first and third results (Theorem~\ref{online:thm:online_ra}
and Theorem~\ref{online:thm:coveringPacking}) are similar to the algorithms 
of~\citet{Y95,Y01} for derandomizing
randomized rounding and the fast approximation algorithms for solving
covering/packing LPs of~\citet*{PST91},~\citet{GK98},~\citet{F00}.  In 
fact~\citet{AHK05} showed that all these algorithms are related to
the multiplicative weights update method for solving the {\em experts} problem
and especially highlighted the similarity between the potential function used in
the analysis of the multiplicative update method and the moment generating
function used in the proof of Chernoff bounds and Young's algorithms.  Hence it
is no surprise that our algorithm which uses Chernoff bounds is also a
multiplicative update algorithm.  Our algorithm is closer in spirit to Young's
algorithms than others.  The main difference is that our algorithm solves an online problem, rather than an offline one, and hence will run short of essential distribution dependent parameters to run the multiplicative weights based algorithm directly: we show that these parameters can be estimated near optimally. And further, we introduce more adversarial models of online input, namely, the varying ASI models, and come up with varying levels of knowledge of the distribution that are sufficient to be able to design good algorithms for these models. And for the offline case, a basic difference of our algorithm from this previous set of results is that our algorithm uses the special structure of the
polytope $\sum_k \xjk \leq 1$ (as against the more general polytopes in these
works) in giving a more efficient solution.  For instance, for our offline problem the number of
oracle calls required will have a quadratic dependence on $\gamma m$ if we used
the~\citet{PST91} algorithm, where as using the special structure of the
polytope, we obtain a linear dependence on $\gamma m$. 

%A basic difference of our algorithm from this previous
%set of results is that in all these works, every single iteration of the
%algorithm involves changing the entire solution vector $x$, while our algorithm
%changes only a single coordinate of the vector $x$ per iteration. In other
%words, our algorithm uses the special structure of the polytope $\pj$ in giving
%a more efficient solution.  
It is possible that our algorithm can also be interpreted as an algorithm for
the experts problem.  In fact~\citet{MSVV05}  asked if there is a
$1-o(1)$ competitive algorithm for Adwords in the i.i.d model with small bid to
budget ratio, and in particular if the algorithms for experts could be used.
They also conjectured that such an algorithm would iteratively adjust a budget
discount factor based on the rate at which the budget is spent.  Our algorithms
for resource allocation problem when specialized for Adwords look exactly like
that, but we do not provide formal connections to the experts framework. 
This was done in follow-up works~\cite{AD15,MG14} which showed that 
essentially the same algorithm as ours can be thought of as using a subroutine 
of Multiplicative-weight updates on a suitably defined learning with experts problem. 

% and with the connections to the experts framework, we answer the questions in~\cite{MSVV05} in the positive.  

\paragraph{Follow-up work} There has been a number of follow-up papers since 
the conference version of this paper has been published.~\citet{AHL12} show 
that for the Adwords problem with a {\em known} distribution, it is enough for 
$\gamma$ to be $O(\epsilon^2)$ to get a $1-\epsilon$ approximation. 
Simultaneously,~\citet{DSA12} showed the same dependence 
of $\gamma = O(\epsilon^2)$ for the Adwords problem, but requiring only a few 
parameters from the distribution.~\citet{KPV13} study a 
generalized version of the adwords problem where, an advertiser's profit, 
instead of being budget-additive, could be an arbitrary submodular function of 
the queries assigned to him. For this problem in the worst case setting, they 
show that no algorithm can obtain better than a $\frac{1}{2}$ approximation, 
which the greedy algorithm already achieves. For the same problem in the i.i.d. 
setting, they show, using techniques we develop in this work, that the greedy 
algorithm obtains a $1-\frac{1}{e}$ approximation. 
~\citet{KRTV14} gave similar guarantees as us, for the random permutation model 
(i.i.d. without replacement), and also get the improved bound of $\gamma = 
O(\epsilon^2)$ for the special case of the Adwords problem. 
%We note that the random permutation model is a special case of the ASI model we introduce, and thus, the improvement in their work lies in not requiring for any distribution dependent parameter at all. While our work also doesn't require any distribution dependent parameter for the unknown i.i.d. model, for the general ASI model it is clear that without asking for any distribution dependent parameter, no algorithm can perform better than in the worst case, and thus we need to ask for a few parameters. 
On the other hand, the algorithms of~\cite{KRTV14} are computationally 
expensive, requiring to solve a linear program for serving every single 
request, where as our algorithm performs a much simpler optimization in every 
step: for the adwords problem for instance, it takes only a linear time to 
perform each step's optimization.  Both~\citet{AD15} and~\citet{MG14} showed 
that essentially the same algorithm as ours also works for the random 
permutation model, with the same guarantees, while also relating it formally to 
the learning from experts framework.~\citet{AD15} also greatly generalize the 
resource allocation framework, to handle arbitrary concave objectives and 
convex constraints.~\citet{ESF14} interpret our algorithm as an exponentiated 
sub-gradient algorithm, show that it works for the random permutation model, 
and give a slight generalization to handle additively separable concave reward 
functions.

\section{Preliminaries \& Main Results}\label{online:sec:prelim}
\subsection{Resource Allocation Framework}\label{online:subsec:ra} We consider
the following framework of optimization problems.  There are $n$ resources, with
resource $i \in \reso$ having a capacity of $c_i$.
%%\footnote{We pick $\reso$ to denote the set of resources to be 
%indicative of an important special case of the resource allocation framework 
%where the set of 
%resources is the set of advertisers. Although $\reso$ could in general be any 
%set, like the edges in a graph in the network routing special case, or the 
%items for sale in the combinatorial auctions special case. See 
%Section~\ref{online:sec:special_cases} for special cases of the framework.}  
There are $m$ requests; 
each request
$j \in \req$ can be satisfied by a vector $\xj \in 
\{0,1\}^K$, with 
coordinates 
$\xjk$,
such that $\sum_k \xjk \leq 1$. Think of vector $\xj$ as picking a single 
option to
satisfy a request from a total of $K$ options. We use $\optn$ to denote
the set of options. The vector $\xj$ consumes $\aij\cdot\xj$
amount of resource $i$, and gives $\wij\cdot\xj$ amount of type $i$ 
profit\footnote{While this notation seems to imply that the number of 
resource-types is equal to the number/set of profit types, namely $n$, this 
choice was made purely to reduce clutter in notation. In general the number/set 
of resource-types could be different from that of the number of profit-types, 
and it's straight-forward to verify that our proofs go through for the general 
case.}. The
$\aij$'s and $\wij$'s are non-negative vectors of length $K$ (and so are the
$\xj$'s). The co-ordinates of the vectors $\aij$ and $\wij$ will be denoted by
$\aijk$ and $\wijk$ respectively, i.e., the $k^{\rm th}$ option consumes $\aijk$
amount of resource $i$ and gives a type $i$ profit of $\wijk$.  The objective is
to maximize the minimum among all types of profit subject to the capacity
constraints on the resources. The following is the linear program relaxation of
the resource allocation problem:
\begin{eqnarray*}\label{online:lp:online_ra}
&&\text{Maximize } \min_{i\in\reso} \sum_{j\in\req}\wij\cdot\xj \text{ s.t.}\\
&& \forall ~i\in\reso, \sum_{j\in\req}\aij\cdot\xj \leq c_i \nonumber\\ 
&& \forall ~j\in\req, \sum_{k\in \optn} \xjk \leq 1 \nonumber \\
&& \forall ~j\in\req,k\in \optn, \xjk \geq 0
\end{eqnarray*}
%Let $\gamma = \max\left(\left\{\frac{\aijk}{c_i}\right\}_{i,j,k} \cup
%\left\{\frac{\wijk}{\opte}\right\}_{i,j,k}\right) $ be the notion corresponding to the 
%bid to budget ratio for Adwords. Here $\opte$ is the optimal offline fractional objective
%value achievable.
%on the expected instance (see
%Section~\ref{online:subsec:online_ra} for a definition of the expected instance). 

%For the ease of exposition in the rest of the paper, we will consider a slightly special 
%case of the resource allocation framework

%We will focus on the canonical case is where each $\pj = \{ \xj \in \Real^K: \sum_k \xjk \leq 1 \}$. 
%This captures the case where there are $K$ discrete options, each with a given profit and consumption. 
%This case captures most of the applications we are interested in, which are
%described in Section \ref{online:sec:special_cases}. 
%All the proofs will be presented for this special case, for ease of exposition. 

%For an example of an application that needs more general polytopes see Section \ref{sec:selective_call_out}. 

Note that dropping a request by not picking any option at all is feasible too. 
For expositional convenience, we will denote not picking any option at all as
having picked the $\bot$ option ($\bot$ may not be in the set $\optn$) for which
$a_{ij\bot} = 0$ and $w_{ij\bot} = 0$ for all $i,j$. 

We consider two versions of the above problem. The first is an online version
with stochastic input: requests are drawn from an unknown distribution.  The
second is an offline problem when the number of requests is much larger than the
number of resources, and our goal is to design a fast PTAS for the problem. 

\subsection{Near-Optimal Online Algorithm for Resource Allocation}\label{online:subsec:online_ra}
We now consider an online version of the resource allocation framework.
Here requests arrive online. We consider the i.i.d.  model, where each request 
is drawn independently from a given distribution. The distribution is unknown to the algorithm. 
The algorithm  knows $m$, the total number of requests. To define our benchmark,
we now define the expected instance. 
%The approximation
%factors we give for resource allocation problems with bounded $\gamma$ use
%a benchmark that is even larger than the expected optimal offline solution benchmark
%used in the denominator of expression~\eqref{intro:eq:PRAS}. Our benchmark is
%the fractional optimal solution of the expected instance, defined below. 
%However, for the 
%adwords problem with unbounded $\gamma$, our competitive ratios are 
%with respect to the optimal integral solution to the distribution instance
%which is incomparable to the expected optimal integral solution. This is elaborated
%further below. 
\paragraph*{Expected Instance} Consider the following {\em expected instance}
of the problem, where everything happens as per expectation.  It is a single
offline instance which is a function of the given distribution over requests and
the total number of requests $m$. Every request in the support of the
distribution is also a request in this instance. The capacities of the resources
in this instance are the same as in the original instance.   Suppose request $j$
has a probability $p_j$ of arriving in the given distribution. 
The resource consumption of $j$ in the expected instance is given
by $m p_j \aij$ for all $i$ and the type $i$ profit is $m p_j \wij$.  The
intuition is that if the requests were drawn from this distribution then the
expected number of times request $j$ is seen is $m p_j$.
%and this is represented in the distribution instance by scaling the consumption and the profit vectors by $m p_j$. 
To summarize, the {\em LP relaxations} of a random instance with set of
requests $R$, and the expected
instance $E$ are as follows (slightly rewritten for convenience). 
%We write the LP relaxation instead of the integer
%program because our benchmark is an upper bound on the expected optimal
%fractional solution. 
%For convenience, we have slightly rewritten the LP when compared
%to the resource allocation LP. 

\begin{table}[H]
\begin{center}
\begin{align}
\label{online:lp:ei}
\text{LP relaxations for random and expected instances}
\end{align} 
\begin{tabular}{l|l}
$\underline{\textbf{Random Instance } \rand}$ & $\underline{\textbf{Expected Instance}  \ E}$ \\\\
Maximize $\lambda \qquad$ s.t. & Maximize $\lambda \qquad$  s.t.\\\\
$\forall ~i\in\reso, \sum_{j\in R,k\in \optn} \wijk\xjk \geq \lambda$ & 
$\forall 
~i\in\reso, \sum_{j\in\req,k\in \optn} mp_j\wijk\xjk \geq \lambda$\\\\
$\forall ~i\in\reso, \sum_{j\in R,k\in \optn} \aijk\xjk \leq {c_i}$ & $\forall 
~i\in\reso, \sum_{j\in\req,k\in \optn}mp_j\aijk\xjk \leq c_i$\\\\ 
$\forall ~j\in R, \sum_{k\in \optn} \xjk \leq 1$ & $\forall ~j\in\req, 
\sum_{k\in 
\optn} \xjk \leq 
1$\\\\
$\forall ~j\in R,k\in \optn, \xjk \geq 0.$  & $\forall ~j\in\req,k\in \optn, 
\xjk 
\geq 
0.$
\end{tabular}
\end{center}
%\caption{LP relaxations for random and expected instances}
%\label{online:lp:ei}
\end{table}

%\begin{eqnarray*}
%&&\text{maximize }\sum_{j \text{ in the support} } m p_j \wj. \xj \text{ s.t.}\\
%&& \forall ~i, \sum_{j} m p_j \aij . \xj \leq {c_i} \\
%&& \forall ~j, \xj \in \pj.
%\end{eqnarray*}
We now prove that the fractional optimal solution to
the expected instance $\opte$ is an upper bound on the expectation of $\optr$, 
where $\optr$ is the offline fractional optimum of the actual sequence of
requests in a random instance $R$.
\begin{lemma}\label{online:lem:benchmark}
$\opte \geq \ex[\optr]$
\end{lemma}
\begin{proof}
The average of optimal solutions for all possible sequences of requests is a feasible solution to the 
expected instance with a profit equal to $\ex[\optr]$. Thus the optimal profit for
the expected instance could only be larger. 
\end{proof}

%We note here that for Lemma~\ref{lem:benchmark} hold it is necessary to allow for 
%fractional solutions to the distribution instance. Thus, to yield competitive ratios
%against the expected value of the optimal fractional solution, our algorithms
%compete with fractional optimal solutions
%to the distribution instance. 
%For the adwords problem with unbounded $\gamma$, 
%discussed in Section~\ref{sec:adwords}, our algorithms compete with optimal integral solutions to the 
%distribution instance, and thus the competitive ratio obtained is against the optimal integral solution to the distribution instance, 
%which is incomparable to the expected value of the optimal integral solution. 

The approximation factor of an algorithm in the i.i.d. model is defined as 
the ratio of the expected profit of the algorithm to the fractional optimal
profit $\opte$ for the expected instance. Let $\gamma = \max\left(\left\{\frac{\aijk}{c_i}\right\}_{i,j,k} \cup
\left\{\frac{\wijk}{\opte}\right\}_{i,j,k}\right)$ be the parameter capturing
the significance of any one request when compared to the total set of requests
that arrive online. The main result is that as $\gamma$ tends to zero, the
approximation factor ratio tends to 1. In fact, we give the almost optimal
trade-off. 

\begin{theorem}\label{online:thm:online_ra}
For any $\epsilon \geq 1/m$, Algorithm~\ref{online:alg:online_ra} achieves an objective value of
$\opte(1-O(\epsilon))$ for the online resource allocation problem with
probability at least $1-\epsilon$, assuming $\gamma =
O(\frac{\epsilon^2}{\log(n/\epsilon)})$.
Algorithm~\ref{online:alg:online_ra} does not require
any knowledge of the distribution at all.
%For any $\epsilon > 0$,  we give an algorithm such that if $\gamma = O\left(\frac{\epsilon^2}{\log(n/\epsilon)}\right)$ then  the competitive ratio of the algorithm is $1 - O(\epsilon)$. 
\end{theorem}

\begin{theorem}\label{online:thm:lb}
There exist instances with $\gamma = \frac{\epsilon^2}{\log (n)}$ such that no
algorithm, even with complete knowledge of the distribution, can get a
$1-o(\epsilon)$ approximation factor. 
%expected objective value of $\opte(1- o(\epsilon))$.
\end{theorem}

\newcommand{\yj}{\mathbf{y}_{j}}

\paragraph*{Oracle Assumption} We assume that we have the following oracle
available to us: given a request $j$ and a vector $\mathbf{v}$, the oracle
returns the vector $\xj$ that maximizes $\mathbf{v}.\xj$ among all $\xj$s in
$\{0,1\}^K$ satisfying $\sum_{k\in\optn} \xjk \leq 1$. 
This assumption boils down to being able to find the
maximum among $K$ numbers, but $K$ may be exponential in some cases. 
%Such an oracle can be implemented
%in polynomial time for most interesting special cases including Adwords, display
%ads, network routing. 
For the Adwords and display ads
problem (described below), $K$ is actually equal to $n$, and this is trivial. 
For network routing (described in
Section~\ref{online:sec:special_cases}), $K$ could be exponential in the size of the network, 
and this assumption corresponds to being
able to find the shortest path in a graph in polynomial time. 
For combinatorial
auctions (described in Section~\ref{online:sec:special_cases}), this corresponds
to the demand query assumption: given prices on various items, the buyer should
be able to decide in polynomial time which bundle gives her the maximum utility.
(While this is not always achievable in polynomial time, there cannot be any hope
of a posted pricing solution for combinatorial auction without this minimum
assumption. )

\paragraph*{Extensions and Special Cases} The extensions of Theorem~\ref{online:thm:online_ra} to the various
generalizations of the i.i.d. model, including the \asi\ model are presented in
Section~\ref{online:sec:asiGWL}. We refer the reader to
Section~\ref{online:sec:special_cases} for a discussion on several problems
that are special cases of the resource allocation framework and have been
previously considered. Here, we discuss two special cases --- the Adwords
problem and display ads problem.

\begin{enumerate}
\item {\bf Adwords.} In the adwords problem, there are $n$ advertisers with a
daily budget of $B_i$ for advertiser $i$. There are $m$ keywords/queries that
arrive online, and advertiser $i$ has a bid of $\bij$ for query $j$. This 
is a special case of the resource allocation framework where the set of options 
$\optn$ matches the set of resources/advertisers 
$\reso$, i.e., each query can be given to at most one advertiser, and will 
consume budget just from that advertiser.  Let $\xij$
denote the indicator variable for whether or not query $j$ was allocated to
agent $i$. After all allocation is over, agent $i$ pays $ 
\min(\sum_{j\in\req}\bij\xij,
B_i)$, i.e., the minimum of the sum of the bids for queries allocated to $i$ and
his budget $B_i$. The objective is to maximize the sum of the payments from all
advertisers --- this is again a special case of the resource allocation 
framework where this only a single profit type, and we just want to maximize 
it. One could raise a technical objection that this is not a special case of 
the resource allocation
framework because the budget constraint is not binding: the value of the
allocated bids to an advertiser can exceed his budget, although the total
payment from the advertiser will be at most the budget. But it is not difficult
to see that the LP relaxation of the offline problem can be written as in
LP~\eqref{online:lp:aDa}, which is clearly a special case of resource allocation
framework LP. Note that the benchmark is anyway an upper
bound even on the expected optimal fractional solution.  Therefore, any
algorithm that gets an $\alpha$ approximation factor for resource allocation is
also guaranteed to get the same approximation factor for Adwords.  The only
notable thing being that an algorithm for resource allocation when used for
adwords will treat the budget constraints as binding, and obtain the guarantee
promised in Theorem~\ref{online:thm:online_ra} (In our $1-1/e$ approximation
algorithm for adwords in Section~\ref{online:sec:adwordsGreedy} that holds
for all values of $\gamma$ ($\leq 1$ of course), we use this facility to exceed
budget).

\item {\bf Display Ads.} In the display ads problem, there are $n$ advertisers
and $m$ impressions arrive online. Advertiser $i$ has wants $c_i$ impressions in
total and pays $\vij$ for impression $j$, and will get paid a penalty of $\rho_i$
for every undelivered impression. If over-delivered, he will pay his bid for the
first $c_i$ impressions delivered. Letting $\bij = \vij +\rho_i$, we can write the
LP relaxation of the offline display ads problem as in LP~\eqref{online:lp:aDa},
which is clearly a special case of the resource allocation
LP, where just like the Adwords special case the set of options $\optn$ is 
equal to 
the set of resources/advertisers $\reso$, and there is only a single profit 
type.  
\end{enumerate}

\begin{table}[H]
\begin{center}
\begin{align}
\label{online:lp:aDa}
\text{LP relaxations for Adwords and Display Ads}
\end{align} 
\begin{tabular}{l|l}
$\underline{\textbf{Adwords}}$ & $\underline{\textbf{Display Ads}}$ \\\\
Maximize $\sum_{i\in \reso,j\in\req}\bij \xij$ s.t. & Maximize 
$\sum_{i\in\reso,j\in\req}\bij\xij$  s.t.\\\\
$\forall ~i\in\reso, \sum_{j\in\req} \bij\xij \leq B_i$ & $\forall ~i\in\reso, 
\sum_{j\in\req}\xij \leq c_i$\\\\
$\forall ~j\in\req, \sum_{i\in\reso} \xij \leq 1$ & $\forall ~j\in\req, 
\sum_{i\in\reso}\xij\leq 1$\\\\ 
$\forall ~i\in\reso,j\in\req, \xij \geq 0.$  & $\forall ~i\in\reso,j\in\req, 
\xij \geq 0.$
\end{tabular}
\end{center}
%\caption{LP relaxations for Adwords and Display Ads}
%\label{online:lp:aDa}
\end{table}

\subsection{Greedy Algorithm for Adwords}\label{online:subsubsec:intro_adwords}
%\newcommand{\btbr}{the bid to budget ratio}
%In the i.i.d {\em Adwords} problem, there are $n$ bidders,  and each bidder $i$ has a daily budget of $B_i$ dollars. 
%{\em Keywords} arrive online with keyword $j$ having an (unknown) probability $p_j$  of arriving in any given step. 
%%A total of $m$ keywords arrive. 
%%We assume that $p_j \geq 1/m$ for each $j$.  
%For every keyword $j$, each bidder submits a bid, $\bij$, which 
%is the profit obtained by the algorithm on allocating keyword $j$ to bidder $i$. 
%The objective is to maximize the profit, subject to the constraint that no bidder is charged more than his budget. 
%Here, the resources are the daily budgets of the bidders, the requests are the keywords, and the options are once again 
%the bidders. The amount of resource consumed and the profit are both $\bij$. 
%For this problem, with no bounds on $\gamma$, 
%we show that the greedy algorithm has a competitive ratio of $1-1/e$.
%For our results for the adwords problem with bounded $\gamma$, see
%Section~\ref{online:subsec:smallAdwords}
As noted in the introduction, the greedy algorithm is widely implemented due to
its simplicity, but its performance was known to be only a $1/2$ approximation
even in stochastic models. We show that the greedy algorithm obtains a $1-1/e$
approximation for all $\gamma$, i.e., $0\leq\gamma\leq 1$. 
\begin{theorem}\label{online:thm:adwords_greedy}
The greedy algorithm achieves an approximation factor of $1-1/e$ 
for the Adwords problem in the i.i.d. unknown distributions model 
for all $\gamma$, i.e., $0\leq\gamma\leq 1$.  
\end{theorem}
We note here that the competitive ratio of $1-1/e$ is tight for the greedy 
algorithm~\cite{GM08}. It is however not known to be tight for an arbitrary algorithm.

\subsection{Fast Approximation Algorithms for Large Mixed Packing and
Covering Integer Programs}\label{online:subsec:intro_covering_packing}

~\citet{CCDJS10} consider the following problem: given a  bipartite graph 
$G=(L,R,E)$
 where $m = |L| \gg |R| =n$, does there exist an assignment $M: L \rightarrow R$ 
with $(j,M(j)) \in E $ for all $j \in L$, and such that for every vertex $i \in R$, $|M^{-1}(i)| \geq B_i$ 
for some given values $B_i$. Since $m$ is very large classic matching algorithms
are not useful.~\citet{CCDJS10} gave an algorithm that runs in time
linear\footnote{In fact, the algorithm makes a single pass through this graph.}
in the number of edges of an induced subgraph obtained by taking a random sample
from $L$ of size $\bigo{\frac{m\log n}{\min_i\{B_i\}\epsilon^2}}$, for a
gap-version of the problem with gap $\epsilon$.  Such an algorithm is very useful in a variety of applications
involving ad assignment for online advertising, particularly when ${\min_i\{B_i\}}$ is
large. 

We consider a generalization of the above problem inspired by the resource
allocation framework.  In fact, we consider the following mixed 
covering-packing integer program. 
Suppose that there are $n$ packing constraints, one for each $i\in [n]$ of the
form $\sum_{j=1}^m \aij \cdot\xj \leq  c_i$ and 
$n$ covering constraints, one for each $i \in [n]$ of the form
$\sum_{j=1}^m \wij \cdot \xj \geq d_i$. 
Each $\xj$ (with coordinates $\xjk$) is constrained to be in $\{0,1\}^K$ and satisfy $\sum_k \xjk \leq
1$. The $\aij$'s and $\wij$'s (and hence
$\xj$'s) are non-negative vectors of length $K$ with coordinates $\aijk$ and
$\wijk$. Does there exist a feasible solution to this system of constraints? 
The gap-version of this problem is as follows. Distinguish between the two cases, with a high probability, say $1-\delta$: 
\begin{itemize}
\item{YES:} There is a feasible solution. 
\item{NO:} There is no feasible solution even if all the $c_i$'s are multiplied by $1+ \epsilon$ and all
the $d_i$'s are multiplied by $1-\epsilon$. 
\end{itemize}
We note that solving (offline) an optimization problem in the resource allocation framework can be reduced to 
the above problem through a binary search on the objective function value. 

%Suppose as in \cite{CCDJS10} that $m$ is much larger than $n$. 
%\paragraph*{Oracle Assumption}
%We assume that we have the following oracle available to us:
%given a request $j$ and a vector $v$, the oracle returns the vector $\xj$ that
%maximizes $v\cdot\xj$ among all $\xj$'s in $\{0, 1\}^K$ satisfying $\sum_k \xjk
%\leq 1$.

%Assume that solving the following oracle call costs unit time: given $j$ and $\mathbf{v}$,
%find $\xj \in \{0,1\}^K$ that maximizes 
%$\mathbf{v}.\xj$.  %Let $\gamma = \max \{ i\in [n_1], j \in [m]: \frac{\aij.\xj}{c_i} \} 
%\cup  \{ i\in [n_2], j \in [m]: \frac{\bbij.\xj}{d_i} \} .$
Let $\gamma = \max\left(\left\{\frac{\aijk}{c_i}\right\}_{i,j,k} \cup
\left\{\frac{\wijk}{d_i}\right\}_{i,j,k}\right)$. 
\begin{theorem}\label{online:thm:coveringPacking}
For any $\epsilon > 0$, assuming $\gamma =
O(\frac{\epsilon^2}{\log(n/\epsilon)})$, Algorithm~\ref{online:alg:offline_ra} 
solves the gap version of the mixed covering-packing integer program  
with $\Theta(\frac{\gamma m \log (n/\delta) }{\epsilon^2})$ oracle calls. 
\end{theorem}

\subsection{Chernoff Bounds} We present here the form of Chernoff bounds that 
we use
throughout the rest of this paper. Let $X = \sum_i X_i$, where $X_i \in [0,B]$ are i.i.d random
variables. Let $\ex[X]=\mu$. Then, for all $\epsilon > 0$, 
$$ \Pr[X < \mu(1-\epsilon)] < \exp\left(\frac{-\epsilon^2\mu}{2B}\right).$$
Consequently, for all $\delta > 0$, with probability at least $1-\delta$, 
$$X -\mu \geq -\sqrt{2\mu B\ln(1/\delta)}.$$
Similarly, for all $\epsilon \in [0, 2e-1]$,
$$ \Pr[X > \mu(1+\epsilon)] < \exp\left(\frac{-\epsilon^2\mu}{4B}\right).$$
Consequently, for all $\delta > \exp(\frac{-(2e-1)^2\mu}{4B})$, with probability at least $1-\delta$, 
$$X -\mu \leq \sqrt{4\mu B\ln(1/\delta)}.$$ 
For $\epsilon > 2e-1$,
$$ \Pr[X > \mu(1+\epsilon)] < 2^{-(1+\epsilon)\mu/B}.$$

%In all our applications of this bound, we ensure
%that $\epsilon \in [0,2e-1]$. 

\section{Near-Optimal Prior Robust Online Algorithms 
for Resource Allocation}\label{online:sec:online}
%\label{sec:online}
For convenience, we begin by rewriting the LP relaxation of a random instance 
$\rand$ of the online resource allocation problem
and the expected instance (already defined in
Section~\ref{online:subsec:online_ra} as LP~\eqref{online:lp:ei}).
\begin{table}[H]
\begin{center}
\begin{align}
\label{online:lp:eiRi}
\text{LPs for random and expected instances}
\end{align}
\begin{tabular}{l|l}
$\underline{\textbf{Random Instance } \rand}$ & $\underline{\textbf{Expected Instance} \ E}$ \\\\
Maximize $\lambda\qquad$ s.t. & Maximize $\lambda\qquad$  s.t.\\\\
$\forall ~i\in \reso, \sum_{j\in R,k\in \optn} \wijk\xjk \geq \lambda$ & 
$\forall 
~i \in 
\reso, 
\sum_{j\in\req,k\in \optn}mp_j\wijk\xjk \geq \lambda$\\\\
$\forall ~i \in \reso, \sum_{j\in R,k \in \optn} \aijk\xjk \leq {c_i}$ & 
$\forall 
~i \in 
\reso, 
\sum_{j\in\req,k\in \optn}mp_j\aijk\xjk \leq c_i$\\\\ 
$\forall ~j\in R, \sum_{k\in \optn} \xjk \leq 1$ & $\forall ~j \in \req, 
\sum_{k\in 
\optn} 
\xjk \leq 
1$\\\\
$\forall ~j\in R,k \in \optn, \xjk \geq 0.$  & $\forall ~j \in \req,k \in 
\optn, 
\xjk 
\geq 0.$
\end{tabular}
\end{center}
%\caption{LPs for random and expected instances}
%\label{online:lp:eiRi}
\end{table}

We showed in Lemma~\ref{online:lem:benchmark} that $\opte \geq \ex[\optr]$. All
our approximation guarantees are w.r.t. the stronger benchmark of $\opte$ which
is the optimal fractional solution of the expected instance. We would like to
remind the reader that while the benchmark is allowed to be fractional, the
online algorithm of course is allowed to find only integral solutions. 

We divide the rest of this section into four subsections. The subsections progressively
weaken the assumptions on knowledge of the distribution of the input. 
\begin{enumerate}
\item In section~\ref{online:sec:HOGWL} we develop a hypothetical algorithm called
Hypothetical-Oblivious-Conservative, denoted by $\hoe$, that achieves an objective value of
$\opte(1-2\epsilon)$ w.p. at least $1-\epsilon$ assuming $\gamma = O\left(\frac{\epsilon^2}{\log(n/\epsilon)}\right)$. 
Theorem~\ref{online:thm:HOGWL} is the main result of this section.
The algorithm is hypothetical because it assumes knowledge of the entire distribution, where as the goal
of this paper is to develop algorithms that work without distributional knowledge. 
\item In section~\ref{online:sec:noDoublingGWL} we design an algorithm for the online resource
allocation problem that achieves the same guarantee as the Hypothetical-Oblivious-Conservative algorithm $\hoe$, 
without any knowledge of the distribution except for a single parameter
of the distribution --- the value of $\opte$. Theorem~\ref{online:thm:noDoublingGWL} is the main result of this section. 
\item In section~\ref{online:sec:doublingGWL} we design an algorithm for the online resource allocation
problem that achieves an objective value of at least $\opte(1-O(\epsilon))$ w.p. at least $1-\epsilon$ 
assuming $\gamma = O(\frac{\epsilon^2}{\log(n/\epsilon)})$ without any knowledge
at all about the distribution. The algorithm in Section~\ref{online:sec:noDoublingGWL}
serves as a good warm-up for the algorithm in this section. Theorem~\ref{online:thm:online_ra} is the main result
of this section. 
\item In section~\ref{online:sec:asiGWL} we relax the assumption that the
distribution from which the requests are drawn is i.i.d.; we give three
different generalizations of the i.i.d. model with strong revenue guarantees
as in the i.i.d. model. 
\end{enumerate}

\subsection{Completely Known Distributions}\label{online:sec:HOGWL}
When the distributions are completely known, we first compute the expected
instance and solve its LP relaxation (LP~\eqref{online:lp:eiRi}) optimally. Let 
$\xjko$
denote the optimal solution to the expected LP~\eqref{online:lp:eiRi}. The 
Hypothetical-Oblivious
algorithm $\ho$ works as follows: when request $j$ arrives, it serves it using
option $k$ with probability $\xjko$. Let $\xito$ denote the amount of resource $i$ consumed in
step $t$ for the algorithm $\ho$.  Thus the total amount of resource $i$
consumed over the entire $m$ steps of algorithm $\ho$ is $\sum_{t=1}^{m} \xito$.
Note that $\ex[\xito] = \sum_{j,k}p_j\aijk\xjko \leq
\frac{c_i}{m}$. Thus, we can bound the probability that
$\Pr[\sum_{t=1}^{m}\xito \geq c_i(1+\epsilon)]$ using Chernoff
bounds. We explicitly derive this bound here since we use this derivation in
designing the algorithm in Section~\ref{online:sec:noDoublingGWL}.

Since we cannot exceed $c_i$ amount of resource consumption by any non-zero amount, we need
to be more conservative than $\ho$. So we analyze
the following algorithm $\hoe$, called Hypothetical-Oblivious-Conservative,
instead of $\ho$: when request $j$ arrives, it serves it using
option $k$ with probability $\frac{\xjko}{1+\epsilon}$, where $\epsilon$ is an error
parameter of algorithm designer's choice. Let $\xitoe$ denote the amount of resource $i$ consumed in
step $t$ for the algorithm $\hoe$. Note that $\ex[\xitoe] \leq
\frac{c_i}{(1+\epsilon)m}$. Thus, even with a $(1+\epsilon)$ deviation using Chernoff bounds,
the resource consumption is at most $c_i$. 

We begin by noting that $\xitoe \leq \gamma c_i$ by the definition of $\gamma$. For all $\epsilon \in [0,1]$ we have,
{\allowdisplaybreaks\begin{align*}
\Pr\bigg[\sum_{t=1}^{m}\xitoe \geq \frac{c_i}{1+\epsilon}(1+\epsilon)\bigg]
&= \Pr\bigg[\frac{\sum_{t=1}^{m}\xitoe}{\gamma c_i} \geq \frac{1}{\gamma}\bigg]\\
& = \Pr\bigg[(1+\epsilon)^{\frac{\sum_{t=1}^{m}\xitoe}{\gamma c_i}} \geq (1+\epsilon)^{\frac{1}{\gamma}}\bigg]\\
&\leq \ex\bigg[(1+\epsilon)^{\frac{\sum_{t=1}^{m}\xitoe}{\gamma c_i}}\bigg]/(1+\epsilon)^{\frac{1}{\gamma}}\\
&= \ex\bigg[\prod_{t=1}^{m}(1+\epsilon)^{\frac{\xitoe}{\gamma c_i}}\bigg]/(1+\epsilon)^{\frac{1}{\gamma}}\\
&\leq \ex\bigg[\prod_{t=1}^{m}\left(1+\epsilon{\frac{\xitoe}{\gamma c_i}}\right)\bigg]/(1+\epsilon)^{\frac{1}{\gamma}}\\
&\leq \bigg[\prod_{t=1}^{m}\left(1+\frac{\epsilon}{(1+\epsilon)\gamma m}\right)\bigg]/(1+\epsilon)^{\frac{1}{\gamma}}\\
&\leq\left(\frac{e^{\epsilon}}{(1+\epsilon)^{1+\epsilon}}\right)^{\frac{1}{\gamma(1+\epsilon)}}\\
&\leq e^{\frac{-\epsilon^2}{4\gamma}\frac{1}{1+\epsilon}}\\
&\leq \frac{\epsilon}{2n}
\end{align*}}
where the first inequality follows from Markov's inequality, the second from convexity of exponential function together with 
with the fact that $\xitoe \leq \gamma c_i$, the third from $\ex[\xitoe]\leq \frac{c_i}{(1+\epsilon)m}$, and the fourth from $1+x \leq e^x$,
the fifth is standard for all $\epsilon \in [0,1]$, and the sixth follows from $\gamma = O(\epsilon^2/\log(n/\epsilon))$ for 
an appropriate choice of constant inside the big-oh coupled with $n \geq 2$. 

\begin{remark}
At first sight this bound might seem anomalous --- the bound $\frac{\epsilon}{2n}$ is increasing in $\epsilon$, i.e., 
the probability of a smaller deviation is smaller than the probability
of a larger deviation! The reason for this anomaly is that $\gamma$ is related to $\epsilon$ as 
$\gamma = O(\frac{\epsilon^2}{\log(n/\epsilon)})$, and smaller the $\gamma$, the better revenue we can get
(i.e., more granular requests leads to lesser wastage from errors, and hence
more revenue). Thus a small deviation for small $\gamma$ has 
a smaller probability than a larger deviation for a larger $\gamma$.
\end{remark}

Similarly let $\yitoe$ denote the revenue obtained from type $i$ profit in step $t$ for the 
algorithm $\hoe$. Note that $\ex[\yitoe] = \sum_{j,k}p_j\wijk\frac{\xjko}{1+\epsilon} \geq
\frac{\opte}{(1+\epsilon)m}$. By the definition of $\gamma$, we have
$\yitoe \leq \gamma\opte$. For all $\epsilon \in [0,1]$ we have,
{\allowdisplaybreaks\begin{align*}
\Pr\bigg[\sum_{t=1}^{m}\yitoe \leq \frac{\opte}{1+\epsilon}(1-\epsilon)\bigg]
&= \Pr\bigg[\frac{\sum_{t=1}^{m}\yitoe}{\gamma \opte} \leq \frac{1-\epsilon}{\gamma(1+\epsilon)}\bigg]\\
& = \Pr\bigg[(1-\epsilon)^{\frac{\sum_{t=1}^{m}\yitoe}{\gamma \opte}} \geq (1-\epsilon)^{\frac{1-\epsilon}{\gamma(1+\epsilon)}}\bigg]\\
&\leq \ex\bigg[(1-\epsilon)^{\frac{\sum_{t=1}^{m}\yitoe}{\gamma \opte}}\bigg]/(1-\epsilon)^{\frac{1-\epsilon}{\gamma(1+\epsilon)}}\\
&= \ex\bigg[\prod_{t=1}^{m}(1-\epsilon)^{\frac{\yitoe}{\gamma \opte}}\bigg]/(1-\epsilon)^{\frac{1-\epsilon}{\gamma(1+\epsilon)}}\\
&\leq \ex\bigg[\prod_{t=1}^{m}\left(1-\epsilon{\frac{\yitoe}{\gamma \opte}}\right)\bigg]/(1-\epsilon)^{\frac{1-\epsilon}{\gamma(1+\epsilon)}}\\
&\leq \bigg[\prod_{t=1}^{m}\left(1-\frac{\epsilon}{(1+\epsilon)\gamma m}\right)\bigg]/(1-\epsilon)^{\frac{1-\epsilon}{\gamma(1+\epsilon)}}\\
&\leq\left(\frac{e^{-\epsilon}}{(1-\epsilon)^{1-\epsilon}}\right)^{\frac{1}{\gamma(1+\epsilon)}}\\
&\leq e^{\frac{-\epsilon^2}{2\gamma}\frac{1}{1+\epsilon}}\\
&\leq \frac{\epsilon}{2n}
\end{align*}}

Thus, we have all the capacity constraints satisfied, (i.e., $\sum_i \xitoe \leq c_i$),
and all resource profits are at least $\frac{\opte}{1+\epsilon}(1-\epsilon)$ 
(i.e., $\sum_i \yitoe \geq \frac{\opte}{1+\epsilon}(1-\epsilon) \geq \opte(1-2\epsilon)$),
with probability at least $1-2n\cdot\epsilon/2n = 1-\epsilon$. This proves the following theorem:
\begin{theorem}\label{online:thm:HOGWL}
For any $\epsilon > 0$, the Hypothetical-Oblivious-Conservative algorithm $\hoe$
achieves an objective value of $\opte(1-2\epsilon)$ for the online resource
allocation problem with probability at least $1-\epsilon$, assuming $\gamma =
O(\frac{\epsilon^2}{\log(n/\epsilon)})$. 
\end{theorem}

\subsection{Unknown Distribution, Known $\opte$}\label{online:sec:noDoublingGWL}
We now design an algorithm\footnote{Note that the notation $\reso$ that we use 
for the set of advertisers/resources is different from the non-calligraphic $A$ 
that we use for an algorithm. Also, it is immediate from the context which one 
we are referring to.} $A$ without knowledge of the distribution, but just 
knowing 
a single parameter $\opte$. Let $A^s\hoe^{m-s}$ be a hybrid
algorithm that runs $A$ for the first $s$ steps and $\hoe$ for the remaining $m-s$ steps. 
Let $\epsilon \in [0,1]$ be the error parameter, which is the algorithm designer's choice. 
Call the algorithm a failure if at least one of the following fails:
\begin{enumerate}
\item For all $i$, $\sum_{t=1}^m \xita \leq c_i$.
%\item For all $i$, $\sum_{t=1}^m \xito \leq c_i(1-\epsilon^2)$.
\item For all $i$, $\sum_{t=1}^m \yita \geq \opte(1-2\epsilon)$.
\end{enumerate}

For any algorithm $A$, let the amount of resource $i$ consumed in the $t$-th step
be denoted by $\xita$ and the amount of resource $i$ profit be denoted by $\yita$. 
Let $\sxia = \sum_{t=1}^s \xita$ denote the amount of resource $i$ consumed in the first
$s$ steps, and let $\syia = \sum_{t=1}^s \yita$ denote the resource $i$ profit
in the first $s$ steps. Similar to the derivation in Section~\ref{online:sec:HOGWL} which bounded
the failure probability of $\hoe$, we can bound the failure probability of any
algorithm $A$, i.e.,

{\allowdisplaybreaks\begin{align}\label{online:eqn:fpNDx}
\Pr\bigg[\sum_{t=1}^{m}\xita \geq \frac{c_i}{1+\epsilon}(1+\epsilon)\bigg]
&= \Pr\bigg[\frac{\sum_{t=1}^{m}\xita}{\gamma c_i} \geq \frac{1}{\gamma}\bigg]\nonumber\\
& = \Pr\bigg[(1+\epsilon)^{\frac{\sum_{t=1}^{m}\xita}{\gamma c_i}} \geq (1+\epsilon)^{\frac{1}{\gamma}}\bigg]\nonumber\\
&\leq \ex\bigg[(1+\epsilon)^{\frac{\sum_{t=1}^{m}\xita}{\gamma c_i}}\bigg]/(1+\epsilon)^{\frac{1}{\gamma}}\nonumber\\
&= \ex\bigg[\prod_{t=1}^{m}(1+\epsilon)^{\frac{\xita}{\gamma c_i}}\bigg]/(1+\epsilon)^{\frac{1}{\gamma}}
\end{align}}
{\allowdisplaybreaks\begin{align}\label{online:eqn:fpNDy}
\Pr\bigg[\sum_{t=1}^{m}\yita \leq \frac{\opte}{1+\epsilon}(1-\epsilon)\bigg]
&= \Pr\bigg[\frac{\sum_{t=1}^{m}\yita}{\gamma \opte} \leq \frac{1-\epsilon}{\gamma(1+\epsilon)}\bigg]\nonumber\\
& = \Pr\bigg[(1-\epsilon)^{\frac{\sum_{t=1}^{m}\yita}{\gamma \opte}} \geq (1-\epsilon)^{\frac{1-\epsilon}{\gamma(1+\epsilon)}}\bigg]\nonumber\\
&\leq \ex\bigg[(1-\epsilon)^{\frac{\sum_{t=1}^{m}\yita}{\gamma \opte}}\bigg]/(1-\epsilon)^{\frac{1-\epsilon}{\gamma(1+\epsilon)}}\nonumber\\
&= \ex\bigg[\prod_{t=1}^{m}(1-\epsilon)^{\frac{\yita}{\gamma \opte}}\bigg]/(1-\epsilon)^{\frac{1-\epsilon}{\gamma(1+\epsilon)}}
\end{align}}
%\end{multicols}

In Section~\ref{online:sec:HOGWL} our algorithm $A$ was $\hoe$ $\bigg($and
therefore we can use $\ex[\xitoe] \leq \frac{c_i}{(1+\epsilon)m}$ and
$\ex[\yitoe] \geq \frac{\opte}{(1+\epsilon)m} \bigg)$, 
the total failure probability which is the sum of~\eqref{online:eqn:fpNDx}
and~\eqref{online:eqn:fpNDy} for all the $i$'s would have been
$n\cdot\left[\frac{\epsilon}{2n} + \frac{\epsilon}{2n}\right] = \epsilon$.
The goal is to design an algorithm $A$ for stage $r$ that, unlike $\ho$, does not know the distribution and knows just $\opte$, 
but obtains the same $\epsilon$ failure probability. That is we want to show
that the sum of~\eqref{online:eqn:fpNDx} and~\eqref{online:eqn:fpNDy} over all $i$'s is at most $\epsilon$: 
%In Section~\ref{sec:HOGWL} we went on to bound the above failure probabilities by $\frac{\epsilon}{2n}$ each, and thus a total failure
%probability of $\epsilon$. Here our goal is design an algorithm $A$ which also has a failure probability of at most $\epsilon$, i.e., we want
$$\sum_i\frac{\ex\bigg[\prod_{t=1}^{m}\left(1+\epsilon\right)^{\frac{\xita}{\gamma c_i}}\bigg]}{(1+\epsilon)^{\frac{1}{\gamma}}}+
\sum_i \frac{\ex\bigg[\prod_{t=1}^{m}\left(1-\epsilon\right)^{\frac{\yita}{\gamma \opte}}\bigg]}{(1-\epsilon)^{\frac{1-\epsilon}{\gamma(1+\epsilon)}}} \leq \epsilon$$
For the algorithm $A^s\hoe^{m-s}$, the above quantity can be rewritten as 
$$\sum_i\frac{\ex\bigg[\left(1+\epsilon\right)^{\frac{\sxia}{\gamma c_i}}
\prod_{t=s+1}^{m}\left(1+\epsilon\right)^{\frac{\xitoe}{\gamma c_i}}\bigg]}{(1+\epsilon)^{\frac{1}{\gamma}}}+
\sum_i \frac{\ex\bigg[\left(1-\epsilon\right)^{\frac{\syia}{\gamma \opte}}
\prod_{t=s+1}^{m}\left(1-\epsilon\right)^{\frac{\yitoe}{\gamma \opte}}\bigg]}{(1-\epsilon)^{\frac{1-\epsilon}{\gamma(1+\epsilon)}}},$$
which, by using $(1+\epsilon)^x \leq 1+\epsilon x$ for $0\leq x\leq 1$, is in turn upper bounded by
\begin{equation*}
\sum_i\frac{\ex\bigg[\left(1+\epsilon\right)^{\frac{\sxia}{\gamma c_i}}
\prod_{t=s+1}^{m}\left(1+\epsilon{\frac{\xitoe}{\gamma c_i}}\right)\bigg]}{(1+\epsilon)^{\frac{1}{\gamma}}}+
\sum_i \frac{\ex\bigg[\left(1-\epsilon\right)^{\frac{\syia}{\gamma \opte}}
\prod_{t=s+1}^{m}\left(1-\epsilon{\frac{\yitoe}{\gamma \opte}}\right)\bigg]}{(1-\epsilon)^{\frac{1-\epsilon}{\gamma(1+\epsilon)}}}.
\end{equation*} 
Since for all $t$, the random variables $\xitoe$, $\xita$, $\yitoe$ and $\yita$ are all independent, and 
$\ex[\xitoe]\leq \frac{c_i}{(1+\epsilon)m}$ and $\ex[\yitoe] \geq \frac{\opte}{(1+\epsilon)m}$, the above is in turn upper bounded by
\begin{equation}\label{online:eqn:FAsPm-s}
\sum_i\frac{\ex\bigg[\left(1+\epsilon\right)^{\frac{\sxia}{\gamma c_i}}
\left(1+\frac{\epsilon}{(1+\epsilon)\gamma m}\right)^{m-s}\bigg]}{(1+\epsilon)^{\frac{1}{\gamma}}}+
\sum_i \frac{\ex\bigg[\left(1-\epsilon\right)^{\frac{\syia}{\gamma \opte}}
\left(1-\frac{\epsilon}{(1+\epsilon)\gamma m}\right)^{m-s}\bigg]}{(1-\epsilon)^{\frac{1-\epsilon}{\gamma(1+\epsilon)}}}.
\end{equation}

Let $\uf[A^s\hoe^{m-s}]$ denote the quantity in	~\eqref{online:eqn:FAsPm-s}, which is
an upper bound on failure probability of the hybrid algorithm $A^s\hoe^{m-s}$.
By Theorem~\ref{online:thm:HOGWL}, we know that $\uf[\hoe^m] \leq \epsilon$. We now prove
that for all $s \in \{0,1,\dots,m-1\}$, $\uf[A^{s+1}\hoe^{m-s-1}] \leq
\uf[A^{s}\hoe^{m-s}]$, thus proving that $\uf[A^m]\leq \epsilon$, i.e., running
the algorithm $A$ for all the $m$ steps results in a failure with probability at
most $\epsilon$. To design such an $A$ we closely follow the derivation of
Chernoff bounds, which is what established that $\uf[\hoe^m]\leq \epsilon$ in
Theorem~\ref{online:thm:HOGWL}.  However the design process will reveal that unlike
algorithm $\hoe$ which needs the entire distribution, just the knowledge of
$\opte$ will do for bounding the failure probability by $\epsilon$. 

Assuming that for all $s < p$, the algorithm $A$ has been defined for the first
$s+1$ steps in such a way that $\uf[A^{s+1}\hoe^{m-s-1}] \leq
\uf[A^s\hoe^{m-s}]$,  we now define $A$ for the $p+1$-th step in a way that
will ensure that $\uf[A^{p+1}\hoe^{m-p-1}] \leq \uf[A^p\hoe^{m-p}]$.  We have
{\allowdisplaybreaks\begin{align}
\uf[A^{p+1}\hoe^{m-p-1}] &= \sum_i \frac{\ex\bigg[(1+\epsilon)^{\frac{\sxia[p+1]}{\gamma c_i}}
\left(1+\frac{\epsilon}{(1+\epsilon)\gamma m}\right)^{m-p-1}\bigg]}{(1+\epsilon)^{\frac{1}{\gamma}}}\qquad+\nonumber\\
&\qquad\qquad\qquad\sum_i \frac{\ex\bigg[(1-\epsilon)^{\frac{\syia[p+1]}{\gamma \opte}}
\left(1-\frac{\epsilon}{(1+\epsilon)\gamma m}\right)^{m-p-1}\bigg]}{(1-\epsilon)^{\frac{1-\epsilon}{\gamma(1+\epsilon)}}}\nonumber\\
&\leq \sum_i \frac{\ex\bigg[(1+\epsilon)^{\frac{\sxia[p]}{\gamma c_i}}\left(1+\epsilon{\frac{X_{i,p+1}^A}{\gamma c_i}}\right)
\left(1+\frac{\epsilon}{(1+\epsilon)\gamma m}\right)^{m-p-1}\bigg]}{(1+\epsilon)^{\frac{1}{\gamma}}}\qquad+\nonumber\\
&\qquad\qquad\qquad\sum_i \frac{\ex\bigg[(1-\epsilon)^{\frac{\syia[p]}{\gamma \opte}}\left(1-\epsilon{\frac{Y_{i,p+1}^A}{\gamma \opte}}\right)
\left(1-\frac{\epsilon}{(1+\epsilon)\gamma
m}\right)^{m-p-1}\bigg]}{(1-\epsilon)^{\frac{1-\epsilon}{\gamma(1+\epsilon)}}}\label{online:eqn:intermediateUF}
\end{align}}
Define
\begin{align*}
\phixi &= \frac{1}{c_i}\bigg[\frac{(1+\epsilon)^{\frac{\sxia}{\gamma c_i}}
\left(1+\frac{\epsilon}{(1+\epsilon)\gamma
m}\right)^{m-s-1}}{(1+\epsilon)^{\frac{1}{\gamma}}}\bigg]\\
\phiyi &= \frac{1}{\opte}\bigg[\frac{(1-\epsilon)^{\frac{\syia}{\gamma \opte}}
\left(1-\frac{\epsilon}{(1+\epsilon)\gamma m}\right)^{m-s-1}}{(1-\epsilon)^{\frac{1-\epsilon}{\gamma(1+\epsilon)}}}\bigg]
\end{align*}
Define the step $p+1$ of algorithm $A$ as picking the following option $k^*$ for request $j$, where:
\begin{equation}\label{online:eqn:algNoDoublingGWL}
k^* = \arg \min_k\left\{\sum_i\aijk\cdot\phixi[p] - \sum_i \wijk\cdot\phiyi[p]\right\}.
\end{equation}
For the sake of clarity, the entire algorithm is presented in
Algorithm~\ref{online:alg:online_ra_opte}. 
\begin{algorithm}[!h]
\caption{: Algorithm for stochastic online resource allocation with unknown distribution, known $\opte$}
\label{online:alg:online_ra_opte}
{\bf Input:} Capacities $c_i$ for $i\in[n]$, the total number of requests $m$, the values of $\gamma$ and
$\opte$, an error parameter $\epsilon >
0.$\\
\textbf{Output:} An online allocation of resources to requests\\
\begin{algorithmic}[1]
%\STATE Choose an error parameter $\epsilon \in (0,1)$
\STATE Initialize $\phixi[0] = \frac{1}{c_i}\bigg[\frac{
\left(1+\frac{\epsilon}{(1+\epsilon)\gamma m}\right)^{m-1}}{(1+\epsilon)^{\frac{1}{\gamma}}}\bigg]$, and, 
$\phiyi[0] = \frac{1}{\opte}\bigg[\frac{
\left(1-\frac{\epsilon}{(1+\epsilon)\gamma m}\right)^{m-1}}{(1-\epsilon)^{\frac{1-\epsilon}{\gamma(1+\epsilon)}}}\bigg]$
\FOR {$s$ = $1$ to $m$}
\STATE If the incoming request is $j$, use the following option $k^*$: 
\begin{equation*}
k^* = \arg \min_{k\in \optn\cup\{\bot\}} \left\{\sum_i\aijk\cdot\phixi[s-1] - 
\sum_i \wijk\cdot\phiyi[s-1]\right\}.
\end{equation*}
\STATE $\xita[s] = a_{ijk^*}$, $\yita[s]=w_{ijk^*}$
\STATE Update $\phixi = \phixi[s-1]\cdot\left[\frac{(1+\epsilon)^{\frac{\xita[s]}{\gamma c_i}}}{1+\frac{\epsilon}{(1+\epsilon)\gamma m}}\right]$, and, 
$\phiyi = \phiyi[s-1]\cdot\left[\frac{(1-\epsilon)^{\frac{\yita[s]}{\gamma\opte}}}{1-\frac{\epsilon}{(1+\epsilon)\gamma m}}\right]$
\ENDFOR
\end{algorithmic}
\end{algorithm}

By the definition of step $p+1$ of algorithm $A$ given in 
equation~\eqref{online:eqn:algNoDoublingGWL}, it follows that for any two
algorithms with the first $p$ steps being identical, and the last $m-p-1$ steps
following the Hypothetical-Oblivious-Conservative algorithm $\hoe$, algorithm
$A$'s $p+1$-th step is the one that minimizes
expression~\eqref{online:eqn:intermediateUF}. In particular it follows that
expression~\eqref{online:eqn:intermediateUF} is upper bounded by the same expression
where the $p+1$-the step is according to $\xitoe[p+1]$ and $\yitoe[p+1]$,
i.e., we replace $\xita[p+1]$ by $\xitoe[p+1]$ and $\yita[p+1]$ by
$\yitoe[p+1]$.  Therefore we have
{\allowdisplaybreaks\begin{align}
\uf[A^{p+1}\hoe^{m-p-1}] &\leq\sum_i \frac{\ex\bigg[(1+\epsilon)^{\frac{\sxia[p]}{\gamma c_i}}\left(1+\epsilon{\frac{\xitoe[p+1]}{\gamma c_i}}\right)
\left(1+\frac{\epsilon}{(1+\epsilon)\gamma m}\right)^{m-p-1}\bigg]}{(1+\epsilon)^{\frac{1}{\gamma}}}\qquad+\nonumber\nonumber\nonumber\\
&\qquad\qquad\qquad\sum_i \frac{\ex\bigg[(1-\epsilon)^{\frac{\syia[p]}{\gamma \opte}}\left(1-\epsilon{\frac{\yitoe[p+1]}{\gamma \opte}}\right)
\left(1-\frac{\epsilon}{(1+\epsilon)\gamma m}\right)^{m-p-1}\bigg]}{(1-\epsilon)^{\frac{1-\epsilon}{\gamma(1+\epsilon)}}}\nonumber\\
&\leq\sum_i \frac{\ex\bigg[(1+\epsilon)^{\frac{\sxia[p]}{\gamma c_i}}\left(1+\frac{\epsilon}{(1+\epsilon)\gamma m}\right)
\left(1+\frac{\epsilon}{(1+\epsilon)\gamma m}\right)^{m-p-1}\bigg]}{(1+\epsilon)^{\frac{1}{\gamma}}}\qquad+\nonumber\nonumber\\
&\qquad\qquad\qquad\sum_i \frac{\ex\bigg[(1-\epsilon)^{\frac{\syia[p]}{\gamma \opte}}\left(1-\frac{\epsilon}{(1+\epsilon)\gamma \opte}\right)
\left(1-\frac{\epsilon}{(1+\epsilon)\gamma m}\right)^{m-p-1}\bigg]}{(1-\epsilon)^{\frac{1-\epsilon}{\gamma(1+\epsilon)}}}\nonumber\\
&=\uf[A^p\hoe^{m-p}]\nonumber
\end{align}}

This completes the proof of the following theorem.
\begin{theorem}\label{online:thm:noDoublingGWL}
For any $\epsilon > 0$, Algorithm~\ref{online:alg:online_ra_opte} 
achieves an objective value of
$\opte(1-2\epsilon)$ for the online resource allocation problem with
probability at least $1-\epsilon$, assuming $\gamma =
O(\frac{\epsilon^2}{\log(n/\epsilon)})$.  The algorithm $A$ does not require
any knowledge of the distribution except for the single parameter $\opte$.
\end{theorem}

\subsection{Completely Unknown Distribution}\label{online:sec:doublingGWL}
We first give a high-level overview of this section before going into the details. In this
section, we design an algorithm $A$ without any knowledge of the distribution at
all. The algorithm is similar in spirit to the one in Section~\ref{online:sec:noDoublingGWL} except that since we do not have knowledge of
$\opte$, we divide the algorithm into many stages.  In each stage, we run an
algorithm similar to the one in Section~\ref{online:sec:noDoublingGWL} except that
instead of $\opte$, we use an estimate of $\opte$ that gets increasingly
accurate with each successive stage. 

More formally, the algorithm runs in $l$ stages $\{0,1,\dots,l-1\}$, where $l$
is such that $\epsilon 2^l = 1$, and $\epsilon \in [1/m,1/2]$ (we need $\epsilon\leq 1/2$ so that
$l$ is at least $1$) is the error
parameter of algorithm designer's choice. Further we need $m\geq\frac{1}{\epsilon}$ so that $\epsilon m \geq 1$. 
We assume that $\epsilon m$ is an integer for clarity of exposition. Stage $\stage$ handles $\tr =
\epsilon m2^{\stage}$ requests for $r\in\{0,\dots l-1\}$. The first $\epsilon m$ requests are used just
for future estimation, and none of them are served. For convenience we
sometimes call this pre-zero stage as stage $-1$, and let $t_{-1} = \epsilon m$.  Stage $r\geq 0$ serves $t
\in [\tr+1,\tr[r+1]]$. Note that in the optimal solution to the expected
instance of stage $\stage$, no resource $i$ gets consumed by more than
$\frac{\tr c_i}{m}$, and every resource $i$ gets a profit of $\frac{\tr
\opte}{m}$, i.e., consumption and profit have been scaled down by a factor of
$\frac{\tr}{m}$. As in the previous sections, with a high probability, we can
only reach close to $\frac{\tr\opte}{m}$.  Further, since stage $r$
consists of only $\tr$ requests, which is much smaller than $m$ for small $r$,
it follows that for small $r$, our error in how close to we get to $\frac{\tr\opte}{m}$ will
be higher. Indeed, instead of having the same error parameter of $\epsilon$ in
every stage, we set stage-specific error parameters which get progressively
smaller and become close to $\epsilon$ in the final stages. These parameters
are chosen such that the overall error is still $O(\epsilon)$ because the later
stages having more requests matter more than the former. There are two
sources of error/failure which we detail below.
\begin{enumerate}
\item The first source of failure stems from not knowing $\opte$. Instead
we estimate a quantity $\zr$ which is an approximation we use for $\opte$ in stage $\stage$, and
the approximation gets better as $\stage$ increases. We use $\zr$ to set a profit target of $\frac{\tr\zr}{m}$ for 
stage $\stage$. Since $\zr$ could be much smaller than $\opte$ our algorithm could become
very suboptimal. We prove that with a probability of at least $1-2\delta$ we have $\opte(1-3\epsx[r-1]) \leq \zr \leq
\opte$ (see next for what $\epsx$ is), where $\delta = \frac{\epsilon}{3l}$. Thus for all the $l$ stages, these bounds are violated with probability
at most $2l\delta = 2\epsilon/3$.
\item The second source of failure stems from achieving this profit target in every stage. 
We set error parameters $\epsx$ and $\epsy$ such that for every $i$
stage $\stage$ consumes at most $\frac{\tr c_i}{m}(1+\epsx)$ amount of resource
$i$, and for every $i$ we get a profit of at least $\frac{\tr
\zr}{m}(1-\epsy)$, with probability at least $1-\delta$. 
Thus the overall failure probability, as regards falling
short of the target $\frac{\tr\zr}{m}$ by more than $\epsy$ and exceeding $\frac{\tr c_i}{m}$ 
by more than $\epsx$, for all the $l$
stages together is at most $\delta\cdot l = \epsilon/3$. 
\end{enumerate}

Thus summing over the failure probabilities we get $\epsilon/3 + 2\epsilon/3 =
\epsilon$.  We have that with probability at least $1-\epsilon$, for every $i$,
the total consumption of resource $i$ is at most
$\sum_{\stage=0}^{l-1}\frac{\tr c_i}{m}(1+\epsx)$, and total profit from
resource $i$ is at least $\sum_{\stage =0}^{l-1}\frac{\tr
\opte}{m}(1-3\epsx[r-1])(1-\epsy)$.  We set $\epsx = \sqrt{\frac{4\gamma
m\ln(2n/\delta)}{\tr}}$ for $r \in \{-1,0,1,\dots,l-1\}$, and $\epsy = \sqrt{\frac{2\wm m
\ln(2n/\delta)}{\tr\zr}}$ for $r \in \{0,1,\dots,l-1\}$ (we define $\epsx$ starting from $r=-1$, with $t_{-1}=\epsilon m$, just for technical convenience).
From this it follows that $\sum_{\stage=0}^{l-1}\frac{\tr c_i}{m}(1+\epsx) \leq c_i$ and $\sum_{\stage =0}^{l-1}\frac{\tr
\opte}{m}(1-3\epsx[r-1])(1-\epsy) \geq \opte(1-O(\epsilon))$, assuming $\gamma =
O(\frac{\epsilon^2}{\log(n/\epsilon)})$. The algorithm is described in
Algorithm~\ref{online:alg:online_ra}. 
This completes the high-level overview of the proof. All that is left to prove is the points 1 and 2 above, upon which
we would have proved our main theorem, namely,
Theorem~\ref{online:thm:online_ra}, which we recall below.
\begin{oneshot}{Theorem~\ref{online:thm:online_ra}}
%\begin{theorem}\label{online:thm:online_ra}
For any $\epsilon \geq 1/m$, Algorithm~\ref{online:alg:online_ra} achieves an objective value of
$\opte(1-O(\epsilon))$ for the online resource allocation problem with
probability at least $1-\epsilon$, assuming $\gamma =
O(\frac{\epsilon^2}{\log(n/\epsilon)})$.
Algorithm~\ref{online:alg:online_ra} does not require
any knowledge of the distribution at all.
%\end{theorem}
\end{oneshot}

\begin{algorithm}[!h]
\caption{: Algorithm for stochastic online resource allocation with unknown distribution}
\label{online:alg:online_ra}
{\bf Input:} Capacities $c_i$ for $i\in[n]$, the total number of requests $m$,
the value of $\gamma$, an error parameter $\epsilon >
1/m.$\\
\textbf{Output:} An online allocation of resources to requests\\

\begin{algorithmic}[1]
\STATE %Choose an error parameter $\epsilon \in [1/m,1/2]$, and 
Set $l = \log(1/\epsilon)$
\STATE Initialize $t_{-1}: t_{-1} \leftarrow \epsilon m$ 
%$X_i^{t_0} \leftarrow 0$, $S_i^{t_0}(0) \leftarrow 0$ \ for\ all $i$, $Y^{t_0} \leftarrow 0$, $V^{t_0}(0) \leftarrow 0$.
\FOR {$\stage$ = 0 to $l-1$}
\STATE Compute $\er[r-1]$: the optimal solution to the $\tr[r-1]$ requests of stage $\stage-1$ with capacities capped at $\frac{\tr[r-1] c_i}{m}$. 
\STATE Set $\zr = \frac{\er[r-1]}{1+\epsx[r-1]}\cdot\frac{m}{\tr[r-1]}$\\
\STATE Set $\epsx = \sqrt{\frac{4\gamma m\ln(2n/\delta)}{\tr}}$, $\epsy = \sqrt{\frac{2\wm m\ln(2n/\delta)}{\tr\zr}}$\\
%\STATE Set $\sxiar[0] = 0$, $\syiar[0]=0$\\
\STATE Set $\phixir[0] = \frac{\epsx}{\gamma c_i}\bigg[\frac{
\left(1+\frac{\epsx}{m\gamma}\right)^{\tr-1}}{(1+\epsx)^{(1+\epsx)\frac{\tr}{m\gamma}}}\bigg]$, and, 
$\phiyir[0] = \frac{\epsy}{\wm}\bigg[\frac{
\left(1-\frac{\epsy}{m\gamma}\right)^{\tr-1}}{(1-\epsy)^{(1-\epsy)\frac{\tr\zr}{m\wm}}}\bigg]$\\
\FOR {$s$ = $1$ to $\tr$}
\STATE If the incoming request is $j$, use the following option $k^*$: 
\begin{equation*}
k^* = \arg \min_{k\in \optn \cup\{\bot\}}\left\{\sum_i\aijk\cdot\phixir[s-1] - 
\sum_i \wijk\cdot\phiyir[s-1]\right\}.
%\arg\min_k \left\{
% \frac{\epsc(\stage)/\gamma}{\left(1 + \frac{\epsc(\stage)}{\gamma m}\right)} \sum_i \phi_i^{init}(\stage) \frac{\aijk}{c_i} 
%  -\ \frac{\epso(\stage)/\wm}{\left(1 - \frac{\epso(\stage) Z(\stage)}{\wm m}\right)} \phi_{\text{obj}}^{init}(\stage) \wjk 
%\right\}.
\end{equation*}
\STATE $\xita[\tr+s] = a_{ijk^*}$, $\yita[\tr+s]=w_{ijk^*}$
%\STATE Update $\sxiar = \sxiar[s-1]+\xita[\tr+s]$, and, $\syiar = \syia[r-1] + \yita[\tr+s]$
\STATE Update $\phixir = \phixir[s-1]\cdot\left[\frac{(1+\epsx)^{\frac{\xita[\tr+s]}{\gamma c_i}}}{1+\frac{\epsx}{m\gamma}}\right]$, and, 
$\phiyir = \phiyir[s-1]\cdot\left[\frac{(1-\epsy)^{\frac{\yita[\tr+s]}{\wm}}}{1-\frac{\epsy}{m\gamma}}\right]$
\ENDFOR
\ENDFOR
\end{algorithmic}
\end{algorithm}

\paragraph*{Detailed Description and Proof}
We begin with the first point in our high-level overview above, namely by describing how $\zr$ is estimated and proving
its concentration around $\opte$. After stage $r$ (including stage $-1$), the algorithm
computes the optimal fractional objective value $\er$ to the following instance $\ir$: the instance
has the $\tr$ requests of stage $\stage$, and the capacity of resource $i$ is
capped at $\frac{\tr c_i}{m}$. Using the optimal fractional objective value $\er$ of this instance,
we set $\zr[r+1] = \frac{\er}{1+\epsx}\cdot\frac{m}{\tr}.$ The first task now is to prove that
$\zr[r+1]$ as estimated above is concentrated enough around $\opte$. This basically requires proving
concentration of $\er$. 

\begin{lemma}\label{online:lem:concOfer}
With a probability at least $1-2\delta$, we have $$ \frac{\tr\opte}{m}(1-2\epsx)\leq \er \leq
\frac{\tr \opte}{m}(1+\epsx).$$
\end{lemma}
\begin{proof}
We prove that the lower and upper bound hold with probability $1-\delta$ each, thus proving the
lemma. 

We begin with the lower bound on $\er$. Note that the expected instance
of the instance $\ir$ has the same optimal solution $\xjko$ as the optimal
solution to the full expected instance (i.e., the one without scaling down by
$\frac{\tr}{m}$). Now consider the algorithm $\hoe(r)$, which is the same as
the $\hoe$ defined in Section~\ref{online:sec:HOGWL} except that $\epsilon$ is
replaced by $\epsx$, i.e., it serves request $j$ with option $k$ with
probability $\frac{\xjko}{1+\epsx}$. When $\hoe(r)$ is run on instance $\ir$,
with a probability at least $1-\frac{\delta}{2n}$, at most $\frac{\tr c_i}{m}$
amount of resource $i$ is consumed, and with probability at least
$1-\frac{\delta}{2n}$, at least $\frac{\tr \opte}{m}\frac{1-\epsx}{1+\epsx}$
resource $i$ profit is obtained. Thus with a probability at least
$1-2n\cdot\frac{\delta}{2n} = 1 -\delta$, $\hoe(r)$ achieves an objective value of at least
$\frac{\tr\opte}{m}(1-2\epsx)$. Therefore the optimal objective value $\er$ will also be at least
$\frac{\tr\opte}{m}(1-2\epsx)$. 

To prove the upper bound, we consider the primal and dual LPs that define $\er$
in LP~\ref{online:lp:dual_er}
and the primal and dual LPs defining the expected instance in 
LP~\eqref{online:lp:dual_ei}. In the latter,
for convenience, we use $mp_j\beta_j$ as the dual multiplier instead of just $\beta_j$. 

\begin{table}[!h]
\begin{center}
\begin{align}
\label{online:lp:dual_er}
\text{Primal and dual LPs defining $\er$}
\end{align}
%\caption{Primal and dual LPs defining $\er$}{
\begin{tabular}{l|l}
\textbf{Primal defining } $\er$ & $\textbf{Dual defining } \er$\\\\
Maximize $\lambda\qquad$ s.t. & Minimize $\sum_{j\in\ir} \beta_j + \frac{\tr}{m}\sum_i \alpha_i c_i\qquad$ s.t.\\\\
$\forall ~i,  ~\sum_{j\in\ir,k} \wijk\xjk \geq \lambda$ & $\forall ~j\in\ir, k, ~\beta_j + \sum_i (\alpha_i\aijk - \rho_i\wijk) \geq 0$\\\\
$\forall ~i, ~\sum_{j\in\ir,k} \aijk\xjk \leq \frac{\tr c_i}{m}$ & $\sum_i \rho_i \geq 1$\\\\
$\forall ~j\in \ir, ~\sum_k \xjk \leq 1$ &  $\forall ~i, ~\rho_i \geq 0, \alpha_i \geq 0$\\\\
$\forall ~j\in \ir,k, ~\xjk \geq 0.$ & $\forall ~j\in\ir, ~\beta_j \geq 0$
\end{tabular}
%\label{online:lp:dual_er}
\end{center}
\end{table}
%%\begin{multicols}{2}
%\begin{align*}
%%\label{lp:primal_er}
%&\textbf{Primal defining } \er\\
%&\text{maximize } \lambda \text{ s.t.}\nonumber\\
%& \forall ~i,  ~\sum_{j\in\ir,k} \wijk\xjk \geq \lambda\nonumber\\
%& \forall ~i, ~\sum_{j\in\ir,k} \aijk\xjk \leq \frac{\tr c_i}{m}\nonumber \\
%& \forall ~j\in \ir, ~\sum_k \xjk \leq 1,\nonumber\\ 
%& \forall ~j,k, ~\xjk \geq 0.\nonumber 
%\end{align*}
%\begin{align}\label{lp:dual_er}
%&\textbf{Dual defining } \er\\
%&\text{minimize } \sum_j \beta_j + \frac{\tr}{m}\sum_i \alpha_i c_i \nonumber\\
%& \forall ~j\in\ir, ~\beta_j + \sum_i (\alpha_i\aijk - \rho_i\wijk) \geq 0\nonumber\\
%& \sum_i \rho_i \leq 1\nonumber\\
%& \forall ~i, ~\rho_i \geq 0, \alpha_i \geq 0\nonumber\\ 
%& \forall ~j\in\ir, ~\beta_j \geq 0\nonumber
%\end{align}
%%\end{multicols}
\begin{table}[!h]
\begin{center}
\begin{align}
\label{online:lp:dual_ei}
\text{Primal and dual LPs defining the expected instance}
\end{align}
\begin{tabular}{l|l}\textbf{Primal for the expected instance} & \textbf{Dual 
for the expected instance}\\\\
Maximize $\lambda\qquad$ s.t. & Minimize $\sum_j mp_j\beta_j + \sum_i \alpha_i c_i\qquad$ s.t.\\\\
$\forall ~i,  ~\sum_{j,k} mp_j\wijk\xjk \geq \lambda$ & $\forall ~j,k, ~mp_j\bigg(\beta_j + \sum_i (\alpha_i\aijk - \rho_i\wijk)\bigg) \geq 0$\\\\
$\forall ~i, ~\sum_{j,k} mp_j\aijk\xjk \leq c_i$ & $\sum_i \rho_i \geq 1$\\\\
$\forall ~j, ~\sum_k \xjk \leq 1$ & $\forall ~i, ~\rho_i \geq 0, \alpha_i \geq 0$\\\\
$\forall ~j,k, ~\xjk \geq 0.$ & $\forall ~j, ~\beta_j \geq 0$\\\\
\end{tabular}
\end{center}
\end{table}
%%\begin{multicols}{2}
%\begin{align*}
%%\label{lp:primal_ei}
%&\textbf{Primal for the expected instance}\\
%&\text{maximize } \lambda \text{ s.t.}\nonumber\\
%& \forall ~i,  ~\sum_{j,k} mp_j\wijk\xjk \geq \lambda\nonumber\\
%& \forall ~i, ~\sum_{j,k} mp_j\aijk\xjk \leq c_i\nonumber \\
%& \forall ~j, ~\sum_k \xjk \leq 1,\nonumber\\ 
%& \forall ~j,k, ~\xjk \geq 0.\nonumber 
%\end{align*}
%\begin{align}\label{lp:dual_ei}
%%\nonumber\\
%%\nonumber\\
%%\nonumber\\
%&\textbf{Dual for the expected instance}\\
%&\text{minimize } \sum_j mp_j\beta_j + \sum_i \alpha_i c_i \nonumber\\
%& \forall ~j\in\ir, ~mp_j\bigg(\beta_j + \sum_i (\alpha_i\aijk - \rho_i\wijk)\bigg) \geq 0\nonumber\\
%& \sum_i \rho_i \leq 1\nonumber\\
%& \forall ~i, ~\rho_i \geq 0, \alpha_i \geq 0\nonumber\\ 
%& \forall ~j, ~\beta_j \geq 0\nonumber
%\end{align}
%%\end{multicols}
Note that the set of constraints in the dual of LP~\eqref{online:lp:dual_ei} is 
a superset of the set of constraints in the dual of 
LP~\eqref{online:lp:dual_er}, making any feasible solution to dual of
LP~\eqref{online:lp:dual_ei} also feasible to dual of 
LP~\eqref{online:lp:dual_er}. 
In particular, the optimal solution to dual of LP~\eqref{online:lp:dual_ei} 
given by $\beta_j^*$'s, $\alpha_i^*$'s and $\rho_i^*$'s is feasible
for dual of LP~\eqref{online:lp:dual_er}. Hence the value of $\er$ is upper
bounded the objective of dual of LP~\eqref{online:lp:dual_er} at $\beta_j^*$'s, 
$\alpha_i^*$'s
and $\rho_i^*$'s. That is we have
$$\er \leq \sum_{j\in\ir} \beta_j^* + \frac{\tr}{m}\sum_i \alpha_i^* c_i.$$
We now upper bound the RHS by applying Chernoff bounds on $\sum_{j\in\ir} \beta_j^*$. 
Since the dual LP in LP~\eqref{online:lp:dual_ei} is a minimization LP, the 
constraints there imply that $\beta_j^* \leq \wm$. 
Applying Chernoff bounds we have,
\begin{eqnarray*}
\er &&\leq \tr\sum_jp_j\beta_j^* + \sqrt{4\tr(\sum_j p_j\beta_j^*)\wm \ln(1/\delta)} + \frac{\tr}{m}\sum_i \alpha_i^* c_i\\
&&\leq\frac{\tr\opte}{m} + \frac{\tr\opte}{m}\epsx
\end{eqnarray*}
where the first inequality holds with probability at least $1-\delta$ and the second inequality uses the fact
the optimal value of the expected instance (dual of 
LP~\eqref{online:lp:dual_ei}) is $\opte$. 
This proves the required upper bound on $\er$ that $\er \leq \frac{\tr\opte}{m}(1+\epsx)$ with probability at least $1-\delta$. 

Going back to our application of Chernoff bounds above, in order to apply it in
the form above, we require that the multiplicative deviation from mean
$\sqrt{\frac{4\wm\ln(1/\delta)}{\tr\sum_j p_j\beta_j^*}} \in [0,2e-1]$. If
$\sum_j p_j \beta_j^* \geq \frac{\epsilon\opte}{m}$, then this requirement
would follow. Suppose on the other hand that $\sum_j p_j \beta_j^* <
\frac{\epsilon\opte}{m}$.  Since we are happy if the excess over mean is at most
$\frac{\tr\opte}{m}\epsx$, let us look for a multiplicative error of
$\frac{\frac{\tr\opte\epsx}{m}}{\tr\sum_jp_j\beta_j^*}$. Based on the fact
that $\sum_j p_j \beta_j^* < \frac{\epsilon\opte}{m}$ and that $\epsx >
\epsilon$ for all $\stage$, the multiplicative error can be seen to be at least
a constant, and can be made larger than $2e-1$ depending on the constant inside
the big O of $\gamma$. We now use the version of Chernoff bounds for
multiplicative error larger than $2e-1$, which gives us that a deviation of
$\frac{\tr\opte}{m}\epsx$ occurs with a probability at most
$2^{-\left(1+\frac{\frac{\tr\opte\epsx}{m}}{\tr\sum_jp_j\beta_j^*}\right)\frac{\tr\sum_jp_j\beta_j^*}{\wm}}$,
where the division by $\wm$ is because of the fact that $\beta_j^* \leq \wm$.
Noting that $\wm \leq \gamma \opte$, we get that this probability is at most
$\delta/n$ which is at most $\delta$.

\end{proof}

Lemma~\ref{online:lem:concOfer} implies that $\opte(1-3\epsx[r-1]) \leq \zr \leq
\opte,\ \forall r\in\{0,1,\dots,l-1\}$. The rest of the proof is similar to
that of Section~\ref{online:sec:noDoublingGWL}, and is focused on the second point in the 
high-level overview we gave in the beginning of this
section~\ref{online:sec:doublingGWL}. In Section~\ref{online:sec:noDoublingGWL}
we knew $\opte$ and obtained a
$\opte(1-2\epsilon)$ approximation with no resource $i$ consumed beyond $c_i$
with probability $1-\epsilon$. Here, instead of $\opte$ we have an
approximation for $\opte$ in the form of $\zr$ which gets increasingly accurate
as $r$ increases. We set a target of $\frac{\tr\zr}{m}$ for stage $r$, and
show that with a probability of at least $1-\delta$ we get a profit of
$\frac{\tr\zr}{m}(1-\epsy)$ from every resource $i$ and no resource $i$
consumed beyond $\frac{\tr c_i}{m}(1+\epsx)$ capacity.\footnote{Note that we are allowed
to consume a bit beyond $\frac{\tr c_i}{m}$ because our goal is just that over
all we don't consume beyond $c_i$, and not that for every stage we respect
the $\frac{\tr c_i}{m}$ constraint. In spite of this $(1+\epsx)$ excess
consumption in all stages, since stage $-1$ consumes nothing at all, we will
see that no excess consumption occurs at the end. 
%Thus since $(1+\epsx)$ excess
%is allowed in every stage, instead of using $\hoe(\stage)$ we use $\ho$, i.e., there is
%no need to be conservative in these stages. Because we are not conservative we are able
%to get a profit of $\frac{\tr \zr}{m}(1-\epsy)$, instead of $\frac{\tr \zr}{m}\frac{1-\epsy}{1+\epsx}$.
}

As in Section~\ref{online:sec:noDoublingGWL}, call stage $r$ of algorithm $A$ a failure if at least one of the following fails:
\begin{enumerate}
\item For all $i$, $\sum_{t=\tr+1}^{\tr[r+1]} \xita \leq \frac{\tr c_i}{m}(1+\epsx)$.
%\item For all $i$, $\sum_{t=1}^m \xito \leq c_i(1-\epsilon^2)$.
\item For all $i$, $\sum_{t=\tr+1}^{\tr[r+1]} \yita \geq \frac{\tr \zr}{m}(1-\epsy)$.
\end{enumerate}

Let $\sxir = \sum_{t=\tr+1}^{\tr+s} \xit$ denote the amount of resource $i$ consumed in the first
$s$ steps of stage $\stage$, and let $\syir = \sum_{t=\tr+1}^{\tr+s} \yit$ denote the resource $i$ profit
in the first $s$ steps of stage $\stage$. 

{\allowdisplaybreaks\begin{align}\label{online:eqn:stagerFpx}
\Pr\bigg[\sum_{t=\tr+1}^{\tr[r+1]}\xita \geq \frac{\tr c_i}{m}(1+\epsx)\bigg]
&= \Pr\bigg[\frac{\sum_{t=\tr+1}^{\tr[r+1]}\xita}{\gamma c_i} \geq \frac{\tr}{m\gamma}(1+\epsx)\bigg]\nonumber\\
& = \Pr\bigg[(1+\epsx)^{\frac{\sum_{t=\tr+1}^{\tr[r+1]}\xita}{\gamma c_i}} \geq (1+\epsx)^{(1+\epsx)\frac{\tr}{m\gamma}}\bigg]\nonumber\\
&\leq \ex\bigg[(1+\epsx)^{\frac{\sum_{t=\tr+1}^{\tr[r+1]}\xita}{\gamma c_i}}\bigg]/(1+\epsx)^{(1+\epsx)\frac{\tr}{m\gamma}}\nonumber\\
&= \frac{\ex\bigg[\prod_{t=\tr+1}^{\tr[r+1]}(1+\epsx)^{\frac{\xita}{\gamma c_i}}\bigg]}{(1+\epsx)^{(1+\epsx)\frac{\tr}{m\gamma}}}
\end{align}}

{\allowdisplaybreaks\begin{align}\label{online:eqn:stagerFpy}
\Pr\bigg[\sum_{t=\tr+1}^{\tr[r+1]}\yita \leq \frac{\tr\zr}{m}(1-\epsy)\bigg]
&= \Pr\bigg[\frac{\sum_{t=\tr+1}^{\tr[r+1]}\yita}{\wm} \leq \frac{\tr\zr}{m\wm}(1-\epsy)\bigg]\nonumber\\
%&\leq \Pr\bigg[\frac{\sum_{t=\tr+1}^{\tr[r+1]}\yita}{\wm} \leq \frac{\tr\zr}{m\wm}(1-\epsy)\bigg]\nonumber\\%(\text{Since Lemma}~\ref{lem:concOfer} \text {implies }\zr\leq \opte \text{ w.h.p.})\\
& = \Pr\bigg[(1-\epsy)^{\frac{\sum_{t=\tr+1}^{\tr[r+1]}\yita}{\wm}} \geq (1-\epsy)^{(1-\epsy)\frac{\tr\zr}{m\wm}}\bigg]\nonumber\\
&\leq \ex\bigg[(1-\epsy)^{\frac{\sum_{t=\tr+1}^{\tr[r+1]}\yita}{\wm}}\bigg]/(1-\epsy)^{(1-\epsy)\frac{\tr\zr}{m\wm}}\nonumber\\
&= \frac{\ex\bigg[\prod_{t=\tr+1}^{\tr[r+1]}(1-\epsy)^{\frac{\yita}{\wm}}\bigg]}{(1-\epsy)^{(1-\epsy)\frac{\tr\zr}{m\wm}}}
\end{align}}

If our algorithm $A$ was $\ho$ $\bigg($and therefore we can use $\ex[\xito] \leq \frac{c_i}{m}$ and $\ex[\yito] \geq \frac{\opte}{m}$ $\geq \frac{\zr}{m}\bigg)$, 
the total failure probability for each stage $\stage$ which is the sum
of~\eqref{online:eqn:stagerFpx} and~\eqref{online:eqn:stagerFpy} for all the $i$'s would have been
$n\cdot\left[e^{\frac{-\epsx^2}{4\gamma}\frac{\tr}{m}} + e^{\frac{-\epsy^2}{2}\frac{\tr\zr}{m\wm}}\right] = n\cdot\left[\frac{\delta}{2n} + \frac{\delta}{2n}\right] = \delta$.
The goal is to design an algorithm $A$ for stage $r$ that, unlike $\ho$, does not know the distribution but also obtains the same $\delta$ failure probability, just as we did
in Section~\ref{online:sec:noDoublingGWL}. That is we want to show that the sum
of~\eqref{online:eqn:stagerFpx} and~\eqref{online:eqn:stagerFpy} over all $i$'s is at most $\delta$: 
%By the derivation on Section~\ref{sec:HOGWL} we know that the failure probability
%of an algorithm that consumes $\xit$ amount of resource $i$ and gives
%a $\yit$ amount of resource $i$ profit at time instance $t$ is upper bounded by
$$\sum_i\frac{\ex\bigg[\prod_{t=\tr+1}^{\tr[r+1]}(1+\epsx)^{\frac{\xita}{\gamma c_i}}\bigg]}{(1+\epsx)^{(1+\epsx)\frac{\tr}{m\gamma}}}+
\sum_i \frac{\ex\bigg[\prod_{t=\tr+1}^{\tr[r+1]}(1-\epsy)^{\frac{\yita}{\wm}}\bigg]}{(1-\epsy)^{(1-\epsy)\frac{\tr\zr}{m\wm}}}
\leq \delta.$$
For the algorithm $A^s\ho^{\tr-s}$, the above quantity can be rewritten as 
\begin{align*}
&\sum_i\frac{\ex\bigg[\left(1+\epsx\right)^{\frac{\sxiar}{\gamma c_i}}
\prod_{t=\tr+s+1}^{\tr[r+1]}\left(1+\epsx\right)^{\frac{\xito}{\gamma
c_i}}\bigg]}{(1+\epsx)^{(1+\epsx)\frac{\tr}{m\gamma}}}+\\
&\qquad\qquad\qquad\sum_i \frac{\ex\bigg[\left(1-\epsy\right)^{\frac{\syiar}{\wm}}
\prod_{t=\tr+s+1}^{\tr[r+1]}\left(1-\epsy\right)^{\frac{\yito}{\wm}}\bigg]}{(1-\epsy)^{(1-\epsy)\frac{\tr\zr}{m\wm}}}
\end{align*}
which, by using $(1+\epsilon)^x \leq 1+\epsilon x$ for $0\leq x\leq 1$, is in turn upper bounded by
\begin{align*}
&\sum_i\frac{\ex\bigg[\left(1+\epsx\right)^{\frac{\sxiar}{\gamma c_i}}
\prod_{t=\tr+s+1}^{\tr[r+1]}\left(1+\epsx{\frac{\xito}{\gamma
c_i}}\right)\bigg]}{(1+\epsx)^{(1+\epsx)\frac{\tr}{m\gamma}}}+\\
&\qquad\qquad\qquad\sum_i \frac{\ex\bigg[\left(1-\epsy\right)^{\frac{\syiar}{\wm}}
\prod_{t=\tr+s+1}^{\tr[r+1]}\left(1-\epsy{\frac{\yito}{\wm}}\right)\bigg]}{(1-\epsy)^{(1-\epsy)\frac{\tr\zr}{m\wm}}}.
\end{align*}

Since for all $t$, the random variables $\xito$, $\xita$, $\yito$ and $\yita$ are all independent, and 
$\ex[\xito]\leq \frac{c_i}{m}$, $\ex[\yito] \geq \frac{\opte}{m}$, and $\frac{\opte}{\wm} \geq \frac{1}{\gamma}$, the above is in turn upper bounded by
\begin{equation}\label{online:eqn:FAsPtr-sr}
\sum_i\frac{\ex\bigg[\left(1+\epsx\right)^{\frac{\sxiar}{\gamma c_i}}
\left(1+\frac{\epsx}{m\gamma}\right)^{\tr-s}\bigg]}{(1+\epsx)^{(1+\epsx)\frac{\tr}{m\gamma}}}+
\sum_i \frac{\ex\bigg[\left(1-\epsy\right)^{\frac{\syiar}{\wm}}
\left(1-\frac{\epsy}{m\gamma}\right)^{\tr-s}\bigg]}{(1-\epsy)^{(1-\epsy)\frac{\tr\zr}{m\wm}}}.
\end{equation}

Let $\ufr[A^s\ho^{\tr-s}]$ denote the quantity in ~\eqref{online:eqn:FAsPtr-sr}, which is
an upper bound on failure probability of the hybrid algorithm $A^s\ho^{\tr-s}$ for stage $\stage$.
We just showed that $\ufr[\ho^{\tr}] \leq \delta$. We now prove
that for all $s \in \{0,1,\dots,\tr-1\}$, $\ufr[A^{s+1}\ho^{\tr-s-1}] \leq
\ufr[A^{s}\ho^{\tr-s}]$, thus proving that $\ufr[A^{\tr}]\leq \delta$, i.e., running
the algorithm $A$ for all the $\tr$ steps of stage $\stage$ results in a failure with probability at
most $\delta$. 
%To design such an $A$ we closely follow the derivation of
%Chernoff bounds, which is what established that $\ufr[\ho^{\tr}]\leq \delta$ in
%Theorem~\ref{thm:HOGWL}.  However the design process will reveal that unlike
%algorithm $\ho$ which needs the entire distribution, just the knowledge of
%$\opte$ will do for bounding the failure probability by $\delta$. 

Assuming that for all $s < p$, the algorithm $A$ has been defined for the first
$s+1$ steps in such a way that $\ufr[A^{s+1}\ho^{\tr-s-1}] \leq
\ufr[A^s\ho^{\tr-s}]$,  we now define $A$ for the $p+1$-th step of stage $\stage$ in a way that
will ensure that $\ufr[A^{p+1}\ho^{\tr-p-1}] \leq \ufr[A^p\ho^{\tr-p}]$.  We have
{\allowdisplaybreaks\begin{align}
\ufr[A^{p+1}\ho^{m-p-1}] &= \sum_i \frac{\ex\bigg[(1+\epsx)^{\frac{\sxiar[p+1]}{\gamma c_i}}
\left(1+\frac{\epsx}{m\gamma}\right)^{\tr-p-1}\bigg]}{(1+\epsx)^{(1+\epsx)\frac{\tr}{m\gamma}}}\qquad+ \nonumber\\
&\qquad\qquad\qquad\sum_i \frac{\ex\bigg[(1-\epsy)^{\frac{\syiar[p+1]}{\wm}}
\left(1-\frac{\epsy}{m\gamma}\right)^{\tr-p-1}\bigg]}{(1-\epsy)^{(1-\epsy)\frac{\tr\zr}{m\wm}}}\nonumber\\
&\leq \sum_i \frac{\ex\bigg[(1+\epsx)^{\frac{\sxiar[p]}{\gamma c_i}}\left(1+\epsx{\frac{X_{i,\tr+p+1}^A}{\gamma c_i}}\right)
\left(1+\frac{\epsx}{m\gamma}\right)^{\tr-p-1}\bigg]}{(1+\epsx)^{(1+\epsx)\frac{\tr}{m\gamma}}}\qquad+\nonumber\\
&\qquad\qquad\qquad\sum_i \frac{\ex\bigg[(1-\epsy)^{\frac{\syiar[p]}{\wm}}\left(1-\epsy{\frac{Y_{i,\tr+p+1}^A}{\wm}}\right)
\left(1-\frac{\epsy}{m\gamma}\right)^{\tr-p-1}\bigg]}{(1-\epsy)^{(1-\epsy)\frac{\tr\zr}{m\wm}}}\label{online:eqn:intermediateUFr}
\end{align}}
Define
\begin{eqnarray*}
\phixir &= \frac{\epsx}{\gamma c_i}\bigg[\frac{(1+\epsx)^{\frac{\sxiar}{\gamma c_i}}
\left(1+\frac{\epsx}{m\gamma}\right)^{\tr-s-1}}{(1+\epsx)^{(1+\epsx)\frac{\tr}{m\gamma}}}\bigg]\label{online:eqn:phixir}\\
\phiyir &= \frac{\epsy}{\wm}\bigg[\frac{(1-\epsy)^{\frac{\syiar}{\wm}}
\left(1-\frac{\epsy}{m\gamma}\right)^{\tr-s-1}}{(1-\epsy)^{(1-\epsy)\frac{\tr\zr}{m\wm}}}\bigg]\label{online:eqn:phixir}
\end{eqnarray*}
Define the step $p+1$ of algorithm $A$ as picking the following option $k^*$ for request $j$:
\begin{equation*}%\label{eqn:algNoDoublingGWL}
k^* = \arg \min_k\left\{\sum_i\aijk\cdot\phixir[p] - \sum_i \wijk\cdot\phiyir[p]\right\}.
\end{equation*}
By the above definition of step $p+1$ of algorithm $A$ (for stage $\stage$), it follows that for any two
algorithms with the first $p$ steps being identical, and the last $\tr-p-1$ steps
following the Hypothetical-Oblivious algorithm $\ho$, algorithm
$A$'s $p+1$-th step is the one that minimizes
expression~\eqref{online:eqn:intermediateUFr}. In particular it follows that
expression~\eqref{online:eqn:intermediateUFr} is upper bounded by the same expression
where the $p+1$-the step is according to $\xito[\tr+p+1]$ and $\yito[\tr+p+1]$,
i.e., we replace $\xita[\tr+p+1]$ by $\xito[\tr+p+1]$ and $\yita[\tr+p+1]$ by
$\yito[\tr+p+1]$.  Therefore we have
{\allowdisplaybreaks\begin{align*}
\ufr[A^{p+1}\ho^{m-p-1}] &\leq\sum_i \frac{\ex\bigg[(1+\epsx)^{\frac{\sxiar[p]}{\gamma c_i}}\left(1+\epsx{\frac{\xito[\tr+p+1]}{\gamma c_i}}\right)
\left(1+\frac{\epsx}{m\gamma}\right)^{\tr-p-1}\bigg]}{(1+\epsx)^{(1+\epsx)\frac{\tr}{m\gamma}}}\qquad+\nonumber\\
&\qquad\qquad\qquad\sum_i \frac{\ex\bigg[(1-\epsy)^{\frac{\syiar[p]}{\wm}}\left(1-\epsy{\frac{\yito[\tr+p+1]}{\wm}}\right)
\left(1-\frac{\epsy}{m\gamma}\right)^{\tr-p-1}\bigg]}{(1-\epsy)^{(1-\epsy)\frac{\tr\zr}{m\wm}}}\nonumber\\
&\leq\sum_i \frac{\ex\bigg[(1+\epsx)^{\frac{\sxiar[p]}{\gamma c_i}}\left(1+\frac{\epsx}{m\gamma}\right)
\left(1+\frac{\epsx}{m\gamma}\right)^{\tr-p-1}\bigg]}{(1+\epsx)^{(1+\epsx)\frac{\tr}{m\gamma}}}\qquad+\nonumber\\
&\qquad\qquad\qquad\sum_i \frac{\ex\bigg[(1-\epsy)^{\frac{\syiar[p]}{\wm}}\left(1-\frac{\epsy}{m\gamma}\right)
\left(1-\frac{\epsy}{m\gamma}\right)^{\tr-p-1}\bigg]}{(1-\epsy)^{(1-\epsy)\frac{\tr\zr}{m\wm}}}\nonumber\\
&=\ufr[A^p\ho^{\tr-p}]\nonumber
\end{align*}}

This completes the proof of Theorem~\ref{online:thm:online_ra}. 

\subsection{Approximate Estimations}\label{online:sec:approxE}
Our Algorithm~\ref{online:alg:online_ra} in Section~\ref{online:sec:doublingGWL}
required periodically computing the optimal solution to an offline instance.
Similarly, our Algorithm~\ref{online:alg:online_ra_opte} in
Section~\ref{online:sec:noDoublingGWL} requires the value of $\opte$ to be
given. Suppose we could only approximately estimate these quantities, do our results
carry through approximately? That is, suppose the solution to the offline
instance is guaranteed to be at least $\frac{1}{\alpha}$ of the optimal, and the
stand-in that we are given for $\opte$ is guaranteed to be at least
$\frac{1}{\alpha}$ of $\opte$. Both our Theorem~\ref{online:thm:online_ra} and
Theorem~\ref{online:thm:noDoublingGWL} go through with just the $\opte$
replaced by $\opte/\alpha$. Every step of the proof of the exact version goes
through in this approximate version, and so we skip the formal proof for this
statement. 

\subsection{Adversarial Stochastic Input}\label{online:sec:asiGWL}
In this section, we relax the assumption that requests are drawn i.i.d. every time step. Namely,  the distribution for each time step need not be identical, but an adversary gets to decide which distribution to sample a request from. The adversary could even use how the algorithm has performed in the first $t-1$ steps in picking the distribution for a given time step $t$. The relevance of this model for the real world is that for settings like display ads, the distribution fluctuates over the day. In general a day is divided into many chunks and within a chunk, the distribution is assumed to remain i.i.d. This is exactly captured by this model. 

We give algorithms that give guarantees against three different benchmarks in the three models below. The benchmarks get successively stronger, and hence the information sought by the algorithm also increases successively.  

\subsubsection{ASI Model 1}

In this model, the guarantee we give is against the worst distribution over all time steps picked by the adversary. More formally, let $\opte(t)$ denote the optimal profit for the expected instance of distribution of time step $t$. Our benchmark will be $\opte = \min_t \opte(t)$. Given just the single number $\opte$, our Algorithm~\ref{online:alg:online_ra_opte} in Section~\ref{online:sec:noDoublingGWL} will guarantee a revenue of $\opte(1-2\epsilon)$
with a probability of at least $1-\epsilon$ assuming $\gamma = O(\frac{\epsilon^2}{\log(n/\epsilon)})$, just
like the guarantee in Theorem~\ref{online:thm:noDoublingGWL}. 

Algorithm~\ref{online:alg:online_ra_opte} works for this ASI model because, the
proof did not use the similarity of the distributions beyond the fact that
$\ex[\xito|\xito[t'] \forall t' < t] \leq \frac{c_i}{m}$ for all values of
$\xito[t']$, and $\ex[\yito|\yito[t'] \forall t' < t] \geq \frac{\opte}{m}$
for all values of $\yito[t']$ (Here $\xito$ and $\yito$ denote the random
variables for resource consumption and profit at time $t$ following from allocation according the optimal solution to the
expected instance of the distribution used in stage $t$).  In other words,
distributions being identical and independent is not crucial, but the fact that
the expected instances of these distributions have a minimum profit guarantee
in spite of all the dependencies between the distributions is sufficient. Both of
these inequalities remain true in this model of ASI also, and thus it easy to
verify that Algorithm~\ref{online:alg:online_ra_opte} works for this model.

\subsubsection{ASI Model 2}

In this model, which is otherwise identical to model 1, our benchmark is stronger, namely, $\opte = \frac{\sum_{t=1}^{m} \opte(t)}{m}$: this is clearly a much stronger benchmark than $\min_t \opte(t)$. Correspondingly, our algorithm requires more information than in model 1: we ask for $\opte(t)$ for every $t$, at the beginning of the algorithm.

A slight modification of our Algorithm~\ref{online:alg:online_ra_opte} in
Section~\ref{online:sec:noDoublingGWL} will give a revenue of
$\frac{\sum_{t=1}^{m} \opte(t)}{m}(1-2\epsilon)$ with probability at least $1-\epsilon$, i.e., $\opte(1-2\epsilon)$ w.p. at least
$(1-\epsilon)$. Among the two potential functions $\phixi$ and $\phiyi$, we modify $\phiyi$ in the most natural way 
to account for the fact that distributions change every step. 

Define
\begin{align*}
\phixi &= \frac{1}{c_i}\bigg[\frac{(1+\epsilon)^{\frac{\sxia}{\gamma c_i}}
\left(1+\frac{\epsilon}{(1+\epsilon)\gamma m}\right)^{m-s-1}}{(1+\epsilon)^{\frac{1}{\gamma}}}\bigg]\\
\phiyi &= \frac{1}{\opte}\bigg[\frac{(1-\epsilon)^{\frac{\syia}{\gamma \opte}}
\prod_{t=s+2}^{m}\left(1-\frac{\epsilon\opte(t)}{(1+\epsilon)\opte\gamma m}\right)}{(1-\epsilon)^{\frac{1-\epsilon}{\gamma(1+\epsilon)}}}\bigg]
\end{align*}
Note that when $\opte(t) = \opte$ for all $t$, then we get precisely the
$\phiyi$ defined in Section~\ref{online:sec:noDoublingGWL} for
Algorithm~\ref{online:alg:online_ra_opte}. 
We present our algorithm below in Algorithm~\ref{online:alg:online_ra_asi}. 

Algorithm~\ref{online:alg:online_ra_asi} works for this ASI model much for the same reason why Algorithm~\ref{online:alg:online_ra_opte} worked for ASI model 2: all the proof needs is that $\ex[\xito|\xito[t'] \forall t' < t] \leq \frac{c_i}{m}$ for all values of
$\xito[t']$, and $\ex[\yito|\yito[t'] \forall t' < t] = \frac{\opte(t)}{m}$
for all values of $\yito[t']$ (Here $\xito$ and $\yito$ denote the random
variables for resource consumption and profit at time $t$ following from allocation according the optimal solution to the
expected instance of the distribution used in stage $t$).

\begin{algorithm}[!h]
\caption{: Algorithm for stochastic online resource allocation in ASI model 2}
\label{online:alg:online_ra_asi}
{\bf Input:} Capacities $c_i$ for $i\in[n]$, the total number of requests $m$, the values of $\gamma$ and
$\opte(t)$ for $t \in [m]$, an error parameter $\epsilon >
0.$\\
\textbf{Output:} An online allocation of resources to requests\\

\begin{algorithmic}[1]
%\STATE $\opte(0) = \opte$
%\STATE Choose an error parameter $\epsilon \in (0,1)$
\STATE Initialize $\phixi[0] = \frac{1}{c_i}\bigg[\frac{
\left(1+\frac{\epsilon}{(1+\epsilon)\gamma m}\right)^{m-1}}{(1+\epsilon)^{\frac{1}{\gamma}}}\bigg]$, and, 
$\phiyi[0] = \frac{1}{\opte}\bigg[\frac{
\prod_{t=2}^m \left(1-\frac{\epsilon\opte(t)}{(1+\epsilon)\opte\gamma m}\right)}{(1-\epsilon)^{\frac{1-\epsilon}{\gamma(1+\epsilon)}}}\bigg]$
\FOR {$s$ = $1$ to $m$}
\STATE If the incoming request is $j$, use the following option $k^*$: 
\begin{equation*}
k^* = \arg \min_{k\in \optn\cup\{\bot\}}\left\{\sum_i\aijk\cdot\phixi[s-1] - 
\sum_i \wijk\cdot\phiyi[s-1]\right\}.
\end{equation*}
\STATE $\xita[s] = a_{ijk^*}$, $\yita[s]=w_{ijk^*}$
\STATE Update $\phixi = \phixi[s-1]\cdot\left[\frac{(1+\epsilon)^{\frac{\xita[s]}{\gamma c_i}}}{1+\frac{\epsilon}{(1+\epsilon)\gamma m}}\right]$, and, 
$\phiyi = \phiyi[s-1]\cdot\left[\frac{(1-\epsilon)^{\frac{\yita[s]}{\gamma\opte}}}{1-\frac{\epsilon}{(1+\epsilon)\gamma m}}\right]$
\ENDFOR
\end{algorithmic}
\end{algorithm}

We skip the proof for the profit guarantee of $\opte(1-2\epsilon)$ since it is
almost identical to the proof in Section~\ref{online:sec:noDoublingGWL} for
Algorithm~\ref{online:alg:online_ra_opte}. 

\subsubsection{ASI Model 3}

In this model, which is otherwise identical to models 1 and 2, our benchmark is 
even stronger: namely, the optimal profit of the expected instance with all the 
time varying distributions (explicitly spelled out in 
LP~\eqref{online:lp:ASIProphet}). This benchmark $\opte$ is the strongest 
benchmark possible. Correspondingly, our algorithm requires more information 
than in model 2: we ask for $\optei(t)$ for every $i$ and $t$, and $c_i(t)$ for 
every $i$ and $t$ at the beginning of the algorithm, where $\optei(t)$ and 
$c_i(t)$ are the amount of type $i$ profit obtained and type $i$ resource 
consumed by the optimal solution to the expected instance in 
LP~\eqref{online:lp:ASIProphet} at step $t$. Namely, $\optei(t) = \sum_{j,k} 
p_{j,t}\wijk\xjkt^*$, and $c_i(t) = \sum_{j,k} p_{j,t}\aijk\xjkt^*$, where 
$\xjkt^*$'s are the optimal solution to LP~\eqref{online:lp:ASIProphet}. 

\begin{table}[!h]
\begin{center}
\begin{align}
\label{online:lp:ASIProphet}
\text{Primal and dual LPs defining the expected instance}
\end{align}
%\caption{Primal and dual LPs defining the expected instance}{
\begin{tabular}{l|l}\textbf{Primal for ASI model 3} & 
\textbf{Dual for ASI model 3}\\\\
Maximize $\lambda\qquad$ s.t. & Minimize $\sum_{j,t} p_{j,t}\beta_j,t + \sum_i \alpha_i c_i\qquad$ s.t.\\\\
$\forall ~i,  ~\sum_{t,j,k} p_{j,t}\wijk\xjkt \geq \lambda$ & $\forall ~j,k, ~p_{j,t}\bigg(\beta_{j,t} + \sum_i (\alpha_i\aijk - \rho_i\wijk)\bigg) \geq 0$\\\\
$\forall ~i, ~\sum_{t,j,k} p_{j,t}\aijk\xjkt \leq c_i$ & $\sum_i \rho_i \geq 1$\\\\
$\forall ~j,t, ~\sum_k \xjkt \leq 1$ & $\forall ~i, ~\rho_i \geq 0, \alpha_i \geq 0$\\\\
$\forall ~j,k,t ~\xjkt \geq 0.$ & $\forall ~j, ~\beta_j \geq 0$\\\\
\end{tabular}
\end{center}
\end{table}

A slight modification of our Algorithm~\ref{online:alg:online_ra_opte} in
Section~\ref{online:sec:noDoublingGWL} will give a revenue of
$\opte(1-2\epsilon)$ with probability at least $1-\epsilon$. We modify the two 
potential functions $\phixi$ and $\phiyi$ in the most natural way  to account 
for the fact that distributions change every step.  Let $\optei = 
\sum_{t=1}^{m} \optei(t)$, and thus, our benchmark $\opte$ is simply $\min_i 
\optei$.  Note also that $\sum_{t=1}^m c_i(t)$, call it $c_i^*$, is at most 
$c_i$ by the feasibility of the optimal solution to 
LP~\eqref{online:lp:ASIProphet}. 

Define
\begin{align*}
%\phixi &= \frac{1}{c_i}\bigg[\frac{(1+\epsilon)^{\frac{\sxia}{\gamma c_i}}
%\left(1+\frac{\epsilon}{(1+\epsilon)\gamma m}\right)^{m-s-1}}{(1+\epsilon)^{\frac{1}{\gamma}}}\bigg]\\
\phixi &= \frac{1}{c_i}\bigg[\frac{(1+\epsilon)^{\frac{\sxia}{\gamma c_i}}
\prod_{t=s+2}^{m}\left(1+\frac{\epsilon c_i(t)}{(1+\epsilon)c_i\gamma }\right)}{(1+\epsilon)^{\frac{1}{\gamma}}}\bigg]\\
\phiyi &= \frac{1}{\optei}\bigg[\frac{(1-\epsilon)^{\frac{\syia}{\gamma \optei}}
\prod_{t=s+2}^{m}\left(1-\frac{\epsilon \optei(t)}{(1+\epsilon)\optei\gamma }\right)}{(1-\epsilon)^{\frac{1-\epsilon}{\gamma(1+\epsilon)}}}\bigg]
\end{align*}
%Note that when $\opte(t) = \opte$ for all $t$, then we get precisely the
%$\phiyi$ defined in Section~\ref{online:sec:noDoublingGWL} for
%Algorithm~\ref{online:alg:online_ra_opte}. 
We present our algorithm below in Algorithm~\ref{online:alg:online_ra_asi_prophet}.

Algorithm~\ref{online:alg:online_ra_asi_prophet} works for this ASI model much for the same reason why Algorithm~\ref{online:alg:online_ra_asi} worked for ASI model 2: all the proof needs is that $\ex[\xito|\xito[t'] \forall t' < t] = \frac{c_i(t)}{m}$ for all values of
$\xito[t']$, and $\ex[\yito|\yito[t'] \forall t' < t] = \frac{\optei(t)}{m}$
for all values of $\yito[t']$ (Here $\xito$ and $\yito$ denote the random
variables for resource consumption and profit at time $t$ following from 
allocation according the optimal solution to the expected instance captured by 
LP~\eqref{online:lp:ASIProphet}).

\begin{algorithm}[!h]
\caption{: Algorithm for stochastic online resource allocation in ASI model 3}
\label{online:alg:online_ra_asi_prophet}
{\bf Input:} Capacities $c_i(t)$ and $c_i$, and profits $\optei(t)$ for $i\in[n],\ $ $t \in [m]$, the total number of requests $m$, the values of $\gamma$ and an error parameter $\epsilon >
0.$\\
\textbf{Output:} An online allocation of resources to requests\\

\begin{algorithmic}[1]
%\STATE $\opte(0) = \opte$
%\STATE Choose an error parameter $\epsilon \in (0,1)$
\STATE Initialize %$\phixi[0] = \frac{1}{c_i}\bigg[\frac{
%\left(1+\frac{\epsilon}{(1+\epsilon)\gamma m}\right)^{m-1}}{(1+\epsilon)^{\frac{1}%{\gamma}}}\bigg]$, and, 
$\phixi[0] = \frac{1}{c_i}\bigg[\frac{\prod_{t=2}^m \left(1+\frac{\epsilon c_i(t)}{(1+\epsilon)c_i\gamma}\right)}{(1+\epsilon)^{\frac{1}{\gamma}}}\bigg]$,
and,
$\phiyi[0] = \frac{1}{\optei}\bigg[\frac{
\prod_{t=2}^m \left(1-\frac{\epsilon \optei(t)}{(1+\epsilon)\optei\gamma}\right)}{(1-\epsilon)^{\frac{1-\epsilon}{\gamma(1+\epsilon)}}}\bigg]$
\FOR {$s$ = $1$ to $m$}
\STATE If the incoming request is $j$, use the following option $k^*$: 
\begin{equation*}
k^* = \arg \min_{k\in \optn\cup\{\bot\}}\left\{\sum_i\aijk\cdot\phixi[s-1] - 
\sum_i \wijk\cdot\phiyi[s-1]\right\}.
\end{equation*}
\STATE $\xita[s] = a_{ijk^*}$, $\yita[s]=w_{ijk^*}$
\STATE Update $\phixi = \phixi[s-1]\cdot\left[\frac{(1+\epsilon)^{\frac{\xita[s]}{\gamma c_i}}}{1+\frac{\epsilon}{(1+\epsilon)\gamma m}}\right]$, and, 
$\phiyi = \phiyi[s-1]\cdot\left[\frac{(1-\epsilon)^{\frac{\yita[s]}{\gamma\optei}}}{1-\frac{\epsilon}{(1+\epsilon)\gamma m}}\right]$
\ENDFOR
\end{algorithmic}
\end{algorithm}

We skip the proof for the profit guarantee of $\opte(1-2\epsilon)$ since it is
almost identical to the proof in Section~\ref{online:sec:noDoublingGWL} for
Algorithm~\ref{online:alg:online_ra_opte}.

\section{Proof of Near-Optimality of Online Algorithm for Resource
Allocation}\label{online:sec:lb}
%\label{sec:lb}
In this section, we construct a family of instances of the resource allocation
problem in the i.i.d. setting for which $\gamma = \omega(\epsilon^2/\log n)$
will rule out a competitive ratio of $1-O(\epsilon)$. The construction closely
follows the  construction by~\citet{AWY09} for proving a similar
result in the random-permutation model. 

The instance has $n = 2^z$ resources with $B$ units of each resource, 
and $Bz(2+1/\alpha) + \sqrt{Bz}$ requests where $\alpha < 1$ is some scalar. Each request has only
one ``option", i.e., each request can either be dropped, or if served, consumes 
the same number of units of a specific subset of resources (which we construct below). This means that a request 
is simply a scalar times a binary string of length $2^z$, with the ones (or the scalars) representing the coordinates
of resources that are consumed by this request, if served. 

The requests are classified into $z$ categories. Each category in expectation
consists of $m/z = B(2+1/\alpha) +\sqrt{B/z}$ requests. A category, indexed by
$i$, is composed of two different binary vectors $v_i$ and $w_i$ (each of length
$2^z$). The easiest way to visualize these vectors is to construct two
$2^z\times z$ $0-1$ matrices, with each matrix consisting of all possible binary
strings of length $z$, written one string in a row. The first matrix lists the
strings in ascending order and the second matrix in descending order. The $i$-th
column of the first matrix multiplied by the scalar $\alpha$ is the vector $v_i$
and the $i$-th column of the second matrix is the vector $w_i$. There are two
properties of these vectors that are useful for us:
\begin{enumerate}
\item The vectors $v_i/\alpha$ and $w_i$ are complements of one another
\item Any matrix of $z$ columns, with column $i$ being either $v_i/\alpha$ 
or $w_i$ has exactly one row with all ones in it.
\end{enumerate}

We are ready to construct the i.i.d. instance now. Each request is drawn from the following
distribution. A given request could be, for each $i$, of type:  

\begin{enumerate}
\item $v_i$ and profit 4$\alpha$ with probability $\frac{B}{\alpha zm}$
\item $w_i$ and profit 3 with probability $\frac{B}{zm}$
\item $w_i$ and profit 2 with probability $\sqrt{\frac{B}{zm}}$
\item $w_i$ and profit 1 with probability $\frac{B}{zm}$
\item Zero vector with zero profit with probability $1 - \frac{2B}{zm} - \sqrt{\frac{B}{zm}} - \frac{B}{\alpha zm} $
\end{enumerate}

We use the following notation for request types: a $(2,w_i)$ request stands for a $w_i$ type request of profit
$2$. Observe that the expected instance has an optimal profit of $\opt = 7B$. This is obtained by picking
for each $i$, the $\frac{B}{\alpha z}$ vectors of type $v_i$ and profit
4$\alpha$, along with $\frac{B}{z}$ vectors of type
$w_i$ with profit 3. Note that this exhausts every unit of every item, and thus, combined
with the fact that the most profitable requests have been served, the value of $7B$ is indeed
the optimal value. This means, any algorithm that obtains a $1-\epsilon$ competitive ratio
must have an expected profit of at least $7B - 7\epsilon B$. 

Let $r_i(w)$ and $r_i(v)$ be the random variables denoting the number of vectors of type $w_i$ and $v_i$ 
picked by some $1-\epsilon$ competitive algorithm $ALG$. Let $a_i(v)$ denote the
total number of vectors of type $v_i$ that arrived in this instance. 

\begin{lemma}\label{lem:noChoiceWithW}
For some constant $k$, the $r_i(w)$'s satisfy
\begin{align*}
\sum_i \ex\bigg[|r_i(w) - B/z|\bigg] \leq 7\epsilon B + 4\sqrt{\alpha k Bz}. 
\end{align*}
\end{lemma}

\begin{proof}
Let $Y$ denote the set of indices $i$ for which $r_i(w) > B/z$. One way to upper
bound the total number of vectors of type $v$ picked by $ALG$ is the following. Split the set of
indices into $Y$ and $X = [z] \setminus Y$. The number of $v$'s from $Y$ is, by
chosen notation, $\sum_{i \in Y}r_i(v)$.  The number of $v$'s from $X$, we show,
is at most $\frac{B - \sum_{i\in Y}r_i(w)}{\alpha}$. Note that since there are
only $B$ copies of every item, it follows that $\alpha[\sum_i r_i(v)] \leq B$,
and $\sum_i r_i(w) \leq B$. Further, by property 2 of $v_i$'s and $w_i$'s, we
have that $\alpha[\sum_{i \in X} r_i(v)] + \sum_{i\in Y} r_i(w) \leq B$. This
means that the number of $v$'s from $X$ is at most $\frac{B - \sum_{i \in Y}
r_i(w)}{\alpha}$. 

Let $P = \sum_{i\in Y} (r_i(w) - B/z)$, and $M = \sum_{i\in X}(B/z - r_i(w))$.
Showing $\ex[P+M]\leq 7\epsilon B + 4\sqrt{\alpha kBz}$ proves the lemma. By an
abuse of notation, let $ALG$ also be
the profit obtained by the algorithm $ALG$ and let best$w_i(t)$ denote the most
profitable $t$ requests of type $w_i$ in a given instance. Note that $4B + \sum_{i=1}^z$best$w_i(B/z) \leq 7B = \opt$. 
We upper-bound $\ex[ALG]$ as:
{\allowdisplaybreaks \begin{align*}
\ex[ALG] &\leq \ex\bigg[\sum_{i=1}^{z} \text{best}w_i(r_i(w))\bigg] + 4\alpha\left[\frac{B - \sum_{i
\in Y} \ex[r_i(w)]}{\alpha} + \sum_{i \in Y} \ex[r_i(v)]\right]\\
&\leq \ex\bigg[\sum_{i=1}^{z} \text{best}_i(B/z) + 3P - M\bigg] + 4\bigg(B -
\ex\bigg[\sum_{i \in
Y}(r_i(w)-B/z) + |Y|B/z\bigg]\bigg)\\
&\qquad\qquad\qquad\qquad + 4\alpha\ex\bigg[\sum_{i \in Y}
r_i(v)\bigg]\\
&\leq \opt -\ex[P+M] + 4\alpha\ex\bigg[\sum_{i \in Y}\bigg(r_i(v) -
\frac{B}{\alpha z}\bigg)\bigg]\\
&\text{ $\bigg($Since $P = \sum_{i\in Y} (r_i(w) - B/z)$ $\bigg)$}\\
&\leq \opt -\ex[P+M] + 4\alpha\ex\bigg[\sum_{i \in Y}\bigg(a_i(v) -
\frac{B}{\alpha z}\bigg)\bigg] (\text{Since $r_i(v) \leq a_i(v)$})\\
&\leq \opt -\ex[P+M] + 4\alpha\ex\bigg[\sum_{i:a_i(v) \geq \frac{B}{\alpha
z}}\bigg(a_i(v) - \frac{B}{\alpha z}\bigg)\bigg]\\
&\leq \opt -\ex[P+M] + 4\alpha\cdot z\cdot k'\cdot\sqrt{\frac{B}{\alpha z}}\\
&\text{   (where $k'$ is some constant from 
Central Limit Theorem)}\\ 
&\leq \opt -\ex[P+M] + 4\sqrt{\alpha kBz} \text{  (where $k$ is $k'^2$)}
\end{align*}}
The inequality that follows from CLT uses the fact that for a random variable $X \sim (m,c/m)$ ($X$ is binomially distributed
with success probability of $c/m$), whenever $c = \omega(1)$, and $c\leq m$, we have that $E[X|X\geq c] = c +k'\sqrt{c}$, for some
constant $k'$. In this case, we have $\frac{B}{\alpha z}$ in place of $c$. For example, if $n = \log(m)$ (and thus $z = \log n = \log\log m$), 
as long as $B = \omega(\log \log m)$ and $B\leq m$, the CLT inequality will hold. Note that $\alpha$ could have been any constant and this argument still holds.
\end{proof}

We are now ready to prove Theorem~\ref{online:thm:lb}, which we restate here for
convenience. 
\begin{oneshot}{Theorem~\ref{online:thm:lb}}
There exist instances with $\gamma = \frac{\epsilon^2}{\log (n)}$ such that no
algorithm, even with complete knowledge of the distribution, can get a $1- o(\epsilon)$
approximation factor.
\end{oneshot}
\begin{proof}
We first give the overview of the proof before providing a detailed argument. 
\paragraph*{Overview.} Lemma~\ref{lem:noChoiceWithW} says that $r_i(w)$ has to be almost always close
to $B/z$ for all $i$. In particular, the probability that $\sum_i |r_i(w) - B/z| \leq 4\left(7\epsilon B + 4\sqrt{\alpha k Bz}\right)$
is at least $3/4$.  In this proof we show, in an argument similar to the one
in~\citet{AWY09}, that if this has to be true, one has to lose a revenue of
$\Omega(\sqrt{Bz}) - 4(7\epsilon B + 4\sqrt{\alpha k Bz})$.  Since $\alpha$ can be
set to any arbitrary constant, this means that we lose a revenue of
$\Omega(\sqrt{Bz})-28\epsilon B$. Since $\opt$ is $7B$, to get a $1-\epsilon$
approximation, we require that $\Omega(\sqrt{Bz})-28\epsilon B \leq
7\epsilon B$. Thus, we need $B \geq \Omega(\frac{\log m}{\epsilon^2})$.  In other
words, we require $\gamma = \frac{1}{B} \leq O(\frac{\epsilon^2}{\log m})$. 

\paragraph*{In Detail} We now proceed to prove the claim that a revenue loss of
$\Omega(\sqrt{Bz}) - 4(7\epsilon B + \sqrt{\alpha k Bz})$ is inevitable. We
just showed that with a probability of at least 3/4, $\sum_i |r_i(w) - B/z|
\leq 4\left(7\epsilon B + 4\sqrt{\alpha k Bz}\right)$.  For now we assume that
$r_i(w)$ should be exactly $B/z$ and later account for the probability $1/4$
leeway and also the $4\left(7\epsilon B + 4\sqrt{\alpha k Bz}\right)$ error
that is allowed by Lemma~\ref{lem:noChoiceWithW}.  With this assumption, we
show that for each $i$ there is a loss of $\Omega(\sqrt{B/z})$.

For each $i$ let $o_i$ denote the number of $(1,w_i)$ requests that the
algorithm served in total.  With a constant probability the number of $3$'s and
$2$'s (of type $w_i$) exceed $B/z$.  If $o_i = \Omega(\sqrt{B/z})$ 
%(for an appropriate small constant inside $\Omega$) 
there is a loss of at least
$\Omega(\sqrt{B/z})$ because of picking ones instead of 2's or 3's. This
establishes the $\Omega(\sqrt{B/z})$ loss that we wanted to prove, for this case.

Suppose $o_i < \Omega(\sqrt{Bz})$. For each $i$, let $R_i$ be the set of requests of type $w_i$ with profit either
1 or 3.  For every $i$, with a constant probability $2B/z - 2\sqrt{B/z} \leq
|R_i | \leq 2B/z + 2\sqrt{B/z}$.
Conditional on the set $R_i$ we make the following two observations:
\begin{itemize} 
\item the types of requests in $R_i$ are independent 
random variables that take value $1$ or $3$ with equal
probability.
\item the order of requests in $R_i$ is a uniformly random permutation of $R_i$
\end{itemize}
Now consider any $(2, w_i)$ request, say $t$-th request, of profit 2. With a constant probability this 
request can be served without violating any capacity constraints, and thus, the algorithm has
to decide whether or not to serve this request. In at least 1/2
of the random permutations of $R_i$, the number of bids from set $R_i$ before
the bid $t$ is less than $B/z$. Conditional on this event, the profits of requests in $R_i$ before $t$,
with a constant probability could:
\begin{enumerate}
\item take values such that there are enough $(3,w_i)$ requests after $t$ to make the total number of $w_i$ requests picked 
by the algorithm to be at least $B/z$;
\item take values  such that even if all the $(3,w_i)$ requests after $t$ were picked, the total number of $w_i$ requests
picked is at most $B/z-\sqrt{B/z}$ with a constant probability. 
\end{enumerate}
%This probability calculation is similar to the one used by Kleinberg [21] in his proof of necessity of condition B ≥ Ω(1/ 2 ).
%For completeness, we derive it in the Lemma C.1 towards the end of the proof.
In the first kind of instances (where number of $(3, w_i)$ requests are more than $B/z$) retaining $(2, w_i)$
causes a loss of 1 as we could have picked a 3 instead. 
In the second kind, skipping $(2, w_i)$ causes a loss of $1$ since we could have picked that $2$ instead of a $1$. 
Thus there is an inevitable constant probability loss of 1 per $(2,w_i)$ request. Thus in expectation, there is a $\Omega(\sqrt{B/z})$ loss. 

Thus whether $o_i =\sqrt{B/z}$ or $o_i < \sqrt{B/z}$, we have established a loss of $\Omega(\sqrt{B/z})$ per $i$ and thus a
total expected loss of $\Omega(\sqrt{Bz})$. This is under the assumption that $r_i(w)$ is exactly $B/z$. There is a leeway of 
$4\left(7\epsilon B + 4\sqrt{\alpha k Bz}\right)$ granted by Lemma~\ref{lem:noChoiceWithW}. Even after that leeway, since 
$\alpha$ can be made an arbitrarily small constant and Lemma~\ref{lem:noChoiceWithW} still holds, we have the loss
at $\Omega(\sqrt{Bz}) - 28\epsilon B$. Now after the leeway, the statement $\sum_i |r_i(w) - B/z| \leq 4\left(7\epsilon B + 4\sqrt{\alpha k Bz}\right)$
has to hold only with probability $3/4$. But even this puts the loss at $\Omega(\sqrt{Bz}) - 21\epsilon B$

Therefore, $E[ALG] \leq \opt - \Omega(\sqrt{Bz}) - 21\epsilon B.$ Since $\opt = 7B$, we have 
$E[ALG] \leq \opt(1 - \Omega(\sqrt{z/B}) - 21\epsilon)$, and in order to get $1 - O(\epsilon)$ approximation we need
$\Omega(\sqrt{z/B} - 21\epsilon) \leq O(\epsilon)$, implying that  $B \geq \Omega(z/\epsilon^2) = \Omega(\log m/\epsilon^2).$
\end{proof}

\section{Greedy Algorithm for Adwords}\label{online:sec:adwordsGreedy}
In this section, we give a simple proof of
Theorem~\ref{online:thm:adwords_greedy}, which we restate below for convenience.
\begin{oneshot}{Theorem~\ref{online:thm:adwords_greedy}}
The greedy algorithm achieves an approximation factor of $1-1/e$ 
for the Adwords problem in the i.i.d. unknown distributions model 
for all $\gamma$, i.e., $0\leq\gamma\leq 1$.  
\end{oneshot}

As noted in Section~\ref{online:subsec:online_ra} where the Adwords problem was
introduced, the budget constraints are not hard, i.e., when a query $j$ arrives,
with a bid amount $b_{ij} > $ remaining budget of $i$, we are still allowed to
allot that query to advertiser $i$, but we only earn a revenue of the remaining
budget of $i$, and not the total value $b_{ij}$.  

~\citet{GM08} prove that the greedy algorithm gives a $(1-1/e)$
approximation to the adwords problem when the queries arrive in a random
permutation or in i.i.d., but under an assumption which almost gets down to
$\gamma$ tending to zero, i.e., bids being much smaller than budgets.  We give a
much simpler proof for a $(1-1/e)$ approximation by greedy algorithm for the
i.i.d. unknown distributions case, and our proof works for all $\gamma$.

Let $p_j$ be the probability of query $j$ appearing in any given impression. 
Let $y_j = mp_j$. 
%We assume $y_j \geq 1$ for all $j$, i.e., in expectation, each query appears at least once. 
Let $x_{ij}$ denote the offline fractional optimal solution for the expected instance. 
Let $w_i(t)$ denote the amount of money spent by advertiser $i$ at time step $t$, i.e., 
for the $t$-th query in the greedy algorithm (to be described below). Let $f_i(0) = \sum_j b_{ij}x_{ij}y_j$. 
Let $f_i(t) = f_i(0) - \sum_{r=1}^{t}w_i(r)$. 
%Let $f_i(t)$ denote the budget remaining for advertiser $i$ after $t$ queries have been served. 
Let $f(t) = \sum_{i=1}^n f_i(t)$. 
%We calculate $f_i(t)$ taking the total budget of $i$ at time zero to be $f_i(0) = \sum_j b_{ij}x_{ij}y_j$,
Note that $f_i(0)$ is the amount spent by $i$ in the offline fractional optimal
solution to the expected instance.

Consider the greedy algorithm which allocates the query $j$ arriving at time $t$ to the 
advertiser who has the maximum effective bid for that query, i.e., $\underset{i}{\mathrm{argmax}}
\min\{b_{ij}, B_i - \sum_{r=1}^{t-1}w_i(r)\}$. We prove that this algorithm obtains a revenue of $(1-1/e)\sum_{i,j} b_{ij}x_{ij}y_j$
and thus gives the desired $1-1/e$ competitive ratio against the fractional
optimal solution to the expected instance. 
%The proof is similar to the proof we presented in Lemma~\ref{lem:convexity} for the resource allocation problem. 
Consider a hypothetical algorithm that allocates queries to advertisers
according to the $x_{ij}$'s. We prove that this hypothetical algorithm obtains an expected revenue of 
$(1-1/e)\sum_{i,j} b_{ij}x_{ij}y_j$, and argue that the greedy algorithm only performs better. 
Let $w^h_i(t)$ and $f^h_i(t)$ denote the quantities analogous to $w_i(t)$ and $f_i(t)$ for the hypothetical algorithm, 
with the initial value $f^h_i(0) = f_i(0) = \sum_j b_{ij}x_{ij}y_j$.
Let $f^h(t) = \sum_{i=1}^{n} f^h_i(t)$. 
Let EXCEED$_i(t)$ denote the set of all $j$ such that $b_{ij}$ is strictly greater
than the remaining budget at the beginning of time step $t$, namely $b_{ij} > B_i - \sum_{r=1}^{t-1}w^h_i(r).$
%Let $e_i(t) = \sum_{j \in \text{EXCEED}_i(t)}b_{ij}x_{ij}y_j$. Thus,
%we have $$\sum_{j \notin \text{EXCEED}_i(t)}b_{ij}x_{ij}y_j = f^h_i(0) - e_i(t).$$
%******************************************************************************************
\begin{lemma}\label{lem:adwordsPrelim}
$\ex[w^h_i(t)|f^h_i(t-1)] \geq  \frac{f^h_i(t-1)}{m}$
\end{lemma}
\begin{proof}
The expected amount amount of money spent at time step $t$, is given by 
\begin{eqnarray}\label{eqn:expenseOfi}
\ex[w^h_i(t)|f^h_i(t-1)] = \underset{j\in\text{EXCEED}_i(t)}{\sum}\bigg(B_i - \sum_{r=1}^{t-1}w^h_i(r)\bigg)\frac{x_{ij}y_j}{m} 
+ \sum_{j\notin \text{EXCEED}_i(t)} b_{ij}\frac{x_{ij}y_j}{m}.
\end{eqnarray}
If $\underset{j\in\text{EXCEED}_i(t)}{\sum}x_{ij}y_j \geq 1$, then by~\eqref{eqn:expenseOfi}, 
$$ \ex[w^h_i(t)|f^h_i(t-1)] \geq \frac{B_i - \sum_{r=1}^{t-1}w^h_i(r)}{m} \geq \frac{f^h_i(0) - \sum_{r=1}^{t-1}w^h_i(r)}{m}
= \frac{f^h_i(t-1)}{m}.$$ 
Suppose on the other hand $\underset{j\in\text{EXCEED}_i(t)}{\sum}x_{ij}y_j < 1.$ 
We can write $\ex[w^h_i(t)|f^h_i(t-1)]$ as
\begin{eqnarray}\label{eqn:adwordsIntCase}
\ex[w^h_i(t)|f^h_i(t-1)] = \frac{f^h_i(0)}{m} -  \underset{j\in\text{EXCEED}_i(t)}{\sum} \bigg(b_{ij} - (B_i - \sum_{r=1}^{t-1}w^h_i(r))\bigg)\frac{x_{ij}y_j}{m}.
\end{eqnarray}
Since $b_{ij} \leq B_i$, and $\underset{j\in\text{EXCEED}_i(t)}{\sum}x_{ij}y_j < 1$,~\eqref{eqn:adwordsIntCase} can be simplified to 
\begin{eqnarray*}
\ex[w^h_i(t)|f^h_i(t-1)] &>& \frac{f^h_i(0)}{m} -  \frac{\sum_{r=1}^{t-1}w^h_i(r)}{m}\\
&=&  \frac{f^h_i(t-1)}{m}.
\end{eqnarray*}
\end{proof}
\begin{lemma}\label{lem:hypo_generalized_adwords}
The hypothetical algorithm satisfies the following: $\ex[f^h(t)|f^h(t-1)] \leq f^h(t-1)(1-1/m)$
\end{lemma}
\begin{proof}
From the definition of $f^h_i(t)$, we have
\begin{eqnarray*}
f^h_i(t) &=& f^h_i(t-1) - w^h_i(t)\\
\ex[f^h_i(t) | f^h_i(t-1)] &=& f^h_i(t-1) - \ex[w^h_i(t)|f^h_i(t-1)] \leq f^h_i(t-1)(1-\frac{1}{m}),
\end{eqnarray*}
where the inequality is due to Lemma~\ref{lem:adwordsPrelim}. Summing over all $i$ gives the Lemma. 
\end{proof}

\begin{lemma}\label{lem:greedy_generalized_adwords}
$\ex[\mathrm{GREEDY}] \geq (1-1/e)\sum_{i,j}b_{ij}x_{ij}y_j$
\end{lemma}
\begin{proof}
Lemma~\ref{lem:hypo_generalized_adwords} proves that for the hypothetical algorithm, 
the value of the difference $f^h(t-1) - \ex[f^h(t)|f^h(t-1)]$, which is the 
expected amount spent at time $t$ by all the advertisers together, conditioned on $f^h(t-1)$,
is at least $\frac{f^h(t-1)}{m}$. But by definition, conditioned on the 
amount of money spent in first $t-1$ steps, the greedy 
algorithm earns the maximum revenue at time step $t$ . Thus, for the
greedy algorithm too, the statement of the lemma~\ref{lem:hypo_generalized_adwords} must hold, namely, 
$\ex[f(t)|f(t-1)] \leq f(t-1)(1-1/m)$.
This means that $\ex[f(m)] \leq f(0)(1-1/m)^m \leq f(0)(1/e)$. Thus the expected 
revenue earned is 
\begin{eqnarray*}
\ex[\sum_{r=1}^{m}w(r)] &=& f(0) -\ex[f(m)]\\
&\geq& f(0)\left(1-1/e\right) \\
&=& \left(1-1/e\right)\sum_{i,j}b_{ij}x_{ij}y_j
\end{eqnarray*}
and this proves the lemma. 
\end{proof}

Lemma~\ref{lem:greedy_generalized_adwords} proves Theorem~\ref{online:thm:adwords_greedy}.

\section{Fast Approximation Algorithm for Large Mixed Packing \& Covering
Integer Programs}\label{online:sec:offline}
%\label{sec:apps}
%\subsection{Mixed Covering-Packing }\label{subsec:packing_covering} 
In this
section, we consider the mixed packing-covering problem stated in
Section~\ref{online:subsec:intro_covering_packing}.  and prove
Theorem~\ref{online:thm:coveringPacking}.  We restate the integer program for 
the mixed
covering-packing problem here.
\begin{align}\label{online:lp:offline_ra}
%&&\text{minimize }\lambda \text{ s.t.}\\
& \forall ~i, \sum_{j,k} \aijk\xjk \leq {c_i}\nonumber \\
& \forall ~i, \sum_{j,k} \wijk\xjk \geq {d_i}\nonumber \\
& \forall ~j, \sum_k \xjk \leq1\nonumber\\ 
& \forall ~j,k, \xjk \in \{0,1\}. 
\end{align}
The goal is to check if there is a feasible solution to this IP. We solve a gap
version of this problem. Distinguish between the two cases with a high
probability, say $1-\delta$: 
\begin{itemize}
\item{YES:} There is a feasible solution. 
\item{NO:} There is no feasible solution even with a slack, namely, even if all 
of the $c_i$'s are
multiplied by $1+3\epsilon(1+\epsilon)$ and all of the $d_i$'s are multiplied by
$1-3\epsilon(1+\epsilon)$. 
\end{itemize}
We use $1+3\epsilon(1+\epsilon)$ and $1-3\epsilon(1+\epsilon)$ for slack 
instead of just 
$1+\epsilon$ and
$1-\epsilon$ purely to reduce notational clutter in what follows (mainly for 
the $\no$ case). 

Like in the online problem, we refer to the quantities indexed by $j$ as
requests, $\aijk$ as resource $i$ consumption, and $\wijk$ as resource $i$ profit, 
and the quantities indexed by $k$ as
options. There are a total of $m$ requests, $n$ resources, and $K$ options, and 
the
``zero'' option denoted by $\bot$.  Recall
that the parameter $\gamma$ for this problem is defined by $\gamma =
\max\left(\left\{\frac{\aijk}{c_i}\right\}_{i,j,k} \cup
\left\{\frac{\wijk}{d_i}\right\}_{i,j,k}\right)$. Our algorithm needs the value
of $m$, $n$ and $\gamma$ (an upper bound on the value of $\gamma$ also
suffices). 

\paragraph*{High-level overview.} We solve this offline problem in an online manner
via random sampling. We sample $T = \Theta(\frac{\gamma m \log (n/\delta)
}{\epsilon^2})$ requests $j$ from the set of possible requests uniformly
at random with replacement, and then design an algorithm that allocates resources online
for these requests. At the end of serving $T$ requests we check if the obtained 
solution proportionally satisfies the constraints of 
IP~\eqref{online:lp:offline_ra}. If yes, we declare YES
as the answer and declare NO otherwise. At the core of the solution is the online sampling 
algorithm we use, which is identical to the techniques used to develop the
online algorithm in Sections~\ref{online:sec:noDoublingGWL}
and~\ref{online:sec:doublingGWL}. We describe our algorithm in
Algorithm~\ref{online:alg:offline_ra}
\begin{algorithm}[!h]
\caption{: Online sampling algorithm for offline mixed covering-packing problems}
\label{online:alg:offline_ra}
\textbf{Input:} The mixed packing and covering IP~\eqref{online:lp:offline_ra},
failure probability $\delta > 0$, and an error parameter $\epsilon > 0$. \\
\textbf{Output:} Distinguish between the cases `YES' where there is a feasible
solution to IP~\eqref{online:lp:offline_ra}, and `NO' where there is no feasible
solution to IP~\eqref{online:lp:offline_ra} even if all the $c_i$'s are
multiplied by $1+3\epsilon(1+\epsilon)$ and all of the $d_i$'s are multiplied by
$1-3\epsilon(1+\epsilon)$. \\

\begin{algorithmic}[1]
\STATE Set $T = \Theta(\frac{\gamma m \log (n/\delta) }{\epsilon^2})$ 
\STATE Initialize $\phixi[0] = \frac{1}{c_i}\bigg[\frac{
\left(1+\frac{\epsilon}{m\gamma}\right)^{T-1}}{(1+\epsilon)^{(1+\epsilon)\frac{T}{m\gamma}}}\bigg]$, and, 
$\phiyi[0] = \frac{1}{d_i}\bigg[\frac{
\left(1-\frac{\epsilon}{m\gamma}\right)^{T-1}}{(1-\epsilon)^{(1-\epsilon)\frac{T}{m\gamma}}}\bigg]$
\FOR {$s$ = $1$ to $T$}
\STATE Sample a request $j$ uniformly at random from the total pool of $m$
requests
\STATE If the incoming request is $j$, use the following option $k^*$: 
\begin{equation*}
k^* = \arg \min_{k\in \optn\cup\{\bot\}}\left\{\sum_i\aijk\cdot\phixi[s-1] - 
\sum_i \wijk\cdot\phiyi[s-1]\right\}.
\end{equation*}
\STATE $\xita[s] = a_{ijk^*}$, $\yita[s]=w_{ijk^*}$
\STATE Update $\phixi = \phixi[s-1]\cdot\left[\frac{(1+\epsilon)^{\frac{\xita[s]}{\gamma c_i}}}{1+\frac{\epsilon}{m\gamma}}\right]$, and, 
$\phiyi = \phiyi[s-1]\cdot\left[\frac{(1-\epsilon)^{\frac{\yita[s]}{\gamma d_i}}}{1-\frac{\epsilon}{m\gamma}}\right]$
\ENDFOR
\IF {$\forall i \sum_{t=1}^{T} \xita < \frac{T c_i}{m}(1+\epsilon)$, and, $\sum_{t=1}^{T} \yita > \frac{Td_i}{m}(1-\epsilon)$}
\STATE Declare YES
\ELSE
\STATE Declare NO
\ENDIF
\end{algorithmic}
\end{algorithm}

The main theorem of this section is Theorem~\ref{online:thm:coveringPacking},
which we restate here:
\begin{oneshot}{Theorem~\ref{online:thm:coveringPacking}}
For any $\epsilon > 0$, Algorithm~\ref{online:alg:offline_ra} solves the gap version of the mixed covering-packing problem  
with $\Theta(\frac{\gamma m \log (n/\delta) }{\epsilon^2})$ oracle calls. 
\end{oneshot}

\paragraph*{Detailed Description and Proof} The proof is in two parts. The first part
proves that our algorithm indeed answers YES when the actual answer is YES with a probability
at least $1-\delta$. The second part is the identical statement for the NO case. 

\paragraph*{The YES case} We begin with the case where the true answer is YES.
Let $\xjko$ denote some feasible solution to the LP relaxation of IP
~\eqref{online:lp:offline_ra}. In a
similar spirit to
Sections~\ref{online:sec:HOGWL},~\ref{online:sec:noDoublingGWL}
and~\ref{online:sec:doublingGWL}, we define the algorithm $\ho$ as follows. It
samples a total of $T=\Theta(\frac{\gamma m \log (n/\delta) }{\epsilon^2})$
requests uniformly at random, with replacement, from the total pool of $m$
requests. When request $j$ is sampled, $\ho$ serves $j$ using option $k$ with
probability $\xjko$. Thus, if we denote by $\xito$ the consumption of resource
$i$ in step $t$ of $\ho$, then we have $\ex[\xito] = \sum_{j=1}^{m}
\frac{1}{m}\sum_k \aijk {\xjko} \leq \frac{c_i}{m}$. This inequality follows
from $\xjko$ being a feasible solution to LP relaxation 
of~\eqref{online:lp:offline_ra}. Similarly let
$\yito$ denote the resource $i$ profit in step $t$ of $\ho$. We have $\ex[\yito]
\geq \frac{d_i}{m}$. We now write the probability that our condition for YES is
violated for some algorithm $A$. 

{\allowdisplaybreaks\begin{align}\label{eqn:offline_Fpx}
\Pr\bigg[\sum_{t=1}^{T}\xita \geq \frac{T c_i}{m}(1+\epsilon)\bigg]
&= \Pr\bigg[\frac{\sum_{t=1}^{T}\xita}{\gamma c_i} \geq \frac{T}{m\gamma}(1+\epsilon)\bigg]\nonumber\\
& = \Pr\bigg[(1+\epsilon)^{\frac{\sum_{t=1}^{T}\xita}{\gamma c_i}} \geq (1+\epsilon)^{(1+\epsilon)\frac{T}{m\gamma}}\bigg]\nonumber\\
&\leq \ex\bigg[(1+\epsilon)^{\frac{\sum_{t=1}^{T}\xita}{\gamma c_i}}\bigg]/(1+\epsilon)^{(1+\epsilon)\frac{T}{m\gamma}}\nonumber\\
&= \frac{\ex\bigg[\prod_{t=1}^{T}(1+\epsilon)^{\frac{\xita}{\gamma c_i}}\bigg]}{(1+\epsilon)^{(1+\epsilon)\frac{T}{m\gamma}}}
\end{align}}

{\allowdisplaybreaks\begin{align}\label{eqn:offline_Fpy}
\Pr\bigg[\sum_{t=1}^{T}\yita \leq \frac{T d_i}{m}(1-\epsilon)\bigg]
&= \Pr\bigg[\frac{\sum_{t=1}^{T}\yita}{\gamma d_i} \geq \frac{T}{m\gamma}(1-\epsilon)\bigg]\nonumber\\
& = \Pr\bigg[(1-\epsilon)^{\frac{\sum_{t=1}^{T}\yita}{\gamma d_i}} \geq (1-\epsilon)^{(1-\epsilon)\frac{T}{m\gamma}}\bigg]\nonumber\\
&\leq \ex\bigg[(1-\epsilon)^{\frac{\sum_{t=1}^{T}\yita}{\gamma d_i}}\bigg]/(1-\epsilon)^{(1-\epsilon)\frac{T}{m\gamma}}\nonumber\\
&= \frac{\ex\bigg[\prod_{t=1}^{T}(1-\epsilon)^{\frac{\yita}{\gamma d_i}}\bigg]}{(1-\epsilon)^{(1-\epsilon)\frac{T}{m\gamma}}}
\end{align}} 

If our algorithm $A$ was $\ho$ $\bigg($and therefore we can use $\ex[\xito] \leq \frac{c_i}{m}$ and $\ex[\yito] \geq \frac{d_i}{m}\bigg)$, 
the total failure probability in the YES case, which is the sum of~\eqref{eqn:offline_Fpx} and~\eqref{eqn:offline_Fpy} for all the $i$'s would have been
at most $\delta$, if $T = \Theta(\frac{\gamma m \log(n/\delta)}{\epsilon^2})$
for an appropriate constant inside $\Theta$. 
The goal is to design an algorithm $A$ that, unlike $\ho$, does not first solve 
LP relaxation of IP~\eqref{online:lp:offline_ra} and then use $\xjko$'s to 
allocate
resources, but allocates online and also obtains the same $\delta$ failure probability, just as we did
in Sections~\ref{online:sec:noDoublingGWL} and~\ref{online:sec:doublingGWL}. That is we want to show that the sum of~\eqref{eqn:offline_Fpx} and~\eqref{eqn:offline_Fpy} over all $i$'s is at most $\delta$: 
%By the derivation on Section~\ref{sec:HOGWL} we know that the failure probability
%of an algorithm that consumes $\xit$ amount of resource $i$ and gives
%a $\yit$ amount of resource $i$ profit at time instance $t$ is upper bounded by
$$\frac{\ex\bigg[\prod_{t=1}^{T}(1+\epsilon)^{\frac{\xita}{\gamma c_i}}\bigg]}{(1+\epsilon)^{(1+\epsilon)\frac{T}{m\gamma}}}
+ \frac{\ex\bigg[\prod_{t=1}^{T}(1-\epsilon)^{\frac{\yita}{\gamma d_i}}\bigg]}{(1-\epsilon)^{(1-\epsilon)\frac{T}{m\gamma}}}
\leq \delta.$$

For the algorithm $A^s\ho^{T-s}$, the above quantity can be rewritten as 
$$\sum_i\frac{\ex\bigg[\left(1+\epsilon\right)^{\frac{\sxia}{\gamma c_i}}
\prod_{t=s+1}^{T}\left(1+\epsilon\right)^{\frac{\xito}{\gamma c_i}}\bigg]}{(1+\epsilon)^{(1+\epsilon)\frac{T}{m\gamma}}}+
\sum_i \frac{\ex\bigg[\left(1-\epsilon\right)^{\frac{\syia}{\gamma d_i}}
\prod_{t=s+1}^{T}\left(1-\epsilon\right)^{\frac{\yito}{\gamma d_i}}\bigg]}{(1-\epsilon)^{(1-\epsilon)\frac{T}{m\gamma}}},$$
which, by using $(1+\epsilon)^x \leq 1+\epsilon x$ for $0\leq x\leq 1$, is in turn upper bounded by
$$\sum_i\frac{\ex\bigg[\left(1+\epsilon\right)^{\frac{\sxia}{\gamma c_i}}
\prod_{t=s+1}^{T}\left(1+\epsilon{\frac{\xito}{\gamma c_i}}\right)\bigg]}{(1+\epsilon)^{(1+\epsilon)\frac{T}{m\gamma}}}+
\sum_i \frac{\ex\bigg[\left(1-\epsilon\right)^{\frac{\syia}{\gamma d_i}}
\prod_{t=s+1}^{T}\left(1-\epsilon{\frac{\yito}{\gamma d_i}}\right)\bigg]}{(1-\epsilon)^{(1-\epsilon)\frac{T}{m\gamma}}}.$$

Since for all $t$, the random variables $\xito$, $\xita$, $\yito$ and $\yita$ are all independent, and 
$\ex[\xito]\leq \frac{c_i}{m}$, $\ex[\yito] \geq \frac{d_i}{m}$, the above is in turn upper bounded by
\begin{equation}\label{eqn:offline_FAsPT-s}
\sum_i\frac{\ex\bigg[\left(1+\epsilon\right)^{\frac{\sxia}{\gamma c_i}}
\left(1+\frac{\epsilon}{m\gamma}\right)^{T-s}\bigg]}{(1+\epsilon)^{(1+\epsilon)\frac{T}{m\gamma}}}+
\sum_i \frac{\ex\bigg[\left(1-\epsilon\right)^{\frac{\syia}{\gamma d_i}}
\left(1-\frac{\epsilon}{m\gamma}\right)^{T-s}\bigg]}{(1-\epsilon)^{(1-\epsilon)\frac{T}{m\gamma}}}.
\end{equation}

Let $\uf[A^s\ho^{T-s}]$ denote the quantity in ~\eqref{eqn:offline_FAsPT-s}, which is
an upper bound on failure probability of the hybrid algorithm $A^s\ho^{T-s}$. 
We just said that $\uf[\ho^{T}] \leq \delta$. We now prove
that for all $s \in \{0,1,\dots,T-1\}$, $\uf[A^{s+1}\ho^{T-s-1}] \leq
\uf[A^{s}\ho^{T-s}]$, thus proving that $\uf[A^{T}]\leq \delta$, i.e., running
the algorithm $A$ for all the $T$ steps of stage $\stage$ results in a failure with probability at
most $\delta$. 
%To design such an $A$ we closely follow the derivation of
%Chernoff bounds, which is what established that $\ufr[\ho^{\tr}]\leq \delta$ in
%Theorem~\ref{thm:HOGWL}.  However the design process will reveal that unlike
%algorithm $\ho$ which needs the entire distribution, just the knowledge of
%$\opte$ will do for bounding the failure probability by $\delta$. 

Assuming that for all $s < p$, the algorithm $A$ has been defined for the first
$s+1$ steps in such a way that $\uf[A^{s+1}\ho^{T-s-1}] \leq
\uf[A^s\ho^{T-s}]$,  we now define $A$ for the $p+1$-th step of stage $\stage$ in a way that
will ensure that $\uf[A^{p+1}\ho^{T-p-1}] \leq \uf[A^p\ho^{T-p}]$.  We have
{\allowdisplaybreaks\begin{align}
\uf[A^{p+1}\ho^{m-p-1}] &= \sum_i \frac{\ex\bigg[(1+\epsilon)^{\frac{\sxia[p+1]}{\gamma c_i}}
\left(1+\frac{\epsilon}{m\gamma}\right)^{T-p-1}\bigg]}{(1+\epsilon)^{(1+\epsilon)\frac{T}{m\gamma}}}\qquad+ \nonumber\\
&\qquad\qquad\qquad\qquad\sum_i \frac{\ex\bigg[(1-\epsilon)^{\frac{\syia[p+1]}{\gamma d_i}}
\left(1-\frac{\epsilon}{m\gamma}\right)^{T-p-1}\bigg]}{(1-\epsilon)^{(1-\epsilon)\frac{T}{m\gamma}}}\nonumber\\
&\leq \sum_i \frac{\ex\bigg[(1+\epsilon)^{\frac{\sxia[p]}{\gamma c_i}}\left(1+\epsilon{\frac{X_{i,p+1}^A}{\gamma c_i}}\right)
\left(1+\frac{\epsilon}{m\gamma}\right)^{T-p-1}\bigg]}{(1+\epsilon)^{(1+\epsilon)\frac{T}{m\gamma}}}\qquad+\nonumber\\
&\qquad\qquad\qquad\qquad\sum_i \frac{\ex\bigg[(1-\epsilon)^{\frac{\syia[p]}{\gamma d_i}}\left(1-\epsilon{\frac{Y_{i,p+1}^A}{\gamma d_i}}\right)
\left(1-\frac{\epsilon}{m\gamma}\right)^{T-p-1}\bigg]}{(1-\epsilon)^{(1-\epsilon)\frac{T}{m\gamma}}}\label{eqn:offline_intermediateUF}
\end{align}}
Define
\begin{eqnarray*}
\phixi = \frac{1}{c_i}\bigg[\frac{(1+\epsilon)^{\frac{\sxia}{\gamma c_i}}
\left(1+\frac{\epsilon}{m\gamma}\right)^{T-s-1}}{(1+\epsilon)^{(1+\epsilon)\frac{T}{m\gamma}}}\bigg]\label{eqn:offline_phixi};
\qquad\phiyi = \frac{1}{d_i}\bigg[\frac{(1-\epsilon)^{\frac{\syia}{\gamma d_i}}
\left(1-\frac{\epsilon}{m\gamma}\right)^{T-s-1}}{(1-\epsilon)^{(1-\epsilon)\frac{T}{m\gamma}}}\bigg]\label{eqn:offline_phixir}
\end{eqnarray*}
Define the step $p+1$ of algorithm $A$ as picking the following option $k$ for request $j$:
\begin{equation*}%\label{eqn:algNoDoublingGWL}
k^* = \arg \min_{k\in K\cup\{\bot\}}\left\{\sum_i\aijk\cdot\phixi[p] - \sum_i \wijk\cdot\phiyi[p]\right\}.
\end{equation*}
By the above definition of step $p+1$ of algorithm $A$, it follows that for any two
algorithms with the first $p$ steps being identical, and the last $T-p-1$ steps
following the Hypothetical-Oblivious algorithm $\ho$, algorithm
$A$'s $p+1$-th step is the one that minimizes
expression~\eqref{eqn:offline_intermediateUF}. In particular it follows that
expression~\eqref{eqn:offline_intermediateUF} is upper bounded by the same expression
where the $p+1$-the step is according to $\xito[p+1]$ and $\yito[p+1]$,
i.e., we replace $\xita[p+1]$ by $\xito[p+1]$ and $\yita[p+1]$ by
$\yito[p+1]$.  Therefore we have
{\allowdisplaybreaks\begin{align*}
\uf[A^{p+1}\ho^{T-p-1}] &\leq\sum_i \frac{\ex\bigg[(1+\epsilon)^{\frac{\sxia[p]}{\gamma c_i}}\left(1+\epsilon{\frac{\xito[p+1]}{\gamma c_i}}\right)
\left(1+\frac{\epsilon}{m\gamma}\right)^{T-p-1}\bigg]}{(1+\epsilon)^{(1+\epsilon)\frac{T}{m\gamma}}}\qquad+\nonumber\\
&\qquad\qquad\qquad\qquad\sum_i \frac{\ex\bigg[(1-\epsilon)^{\frac{\syia[p]}{\gamma d_i}}\left(1-\epsilon{\frac{\yito[T+p+1]}{\gamma d_i}}\right)
\left(1-\frac{\epsilon}{m\gamma}\right)^{T-p-1}\bigg]}{(1-\epsilon)^{(1-\epsilon)\frac{T}{m\gamma}}}\nonumber\\
&\leq\sum_i \frac{\ex\bigg[(1+\epsilon)^{\frac{\sxia[p]}{\gamma c_i}}\left(1+\frac{\epsilon}{m\gamma}\right)
\left(1+\frac{\epsilon}{m\gamma}\right)^{T-p-1}\bigg]}{(1+\epsilon)^{(1+\epsilon)\frac{T}{m\gamma}}}\qquad+\nonumber\\
&\qquad\qquad\qquad\qquad\sum_i \frac{\ex\bigg[(1-\epsilon)^{\frac{\syia[p]}{\gamma d_i}}\left(1-\frac{\epsilon}{m\gamma}\right)
\left(1-\frac{\epsilon}{m\gamma}\right)^{T-p-1}\bigg]}{(1-\epsilon)^{(1-\epsilon)\frac{T}{m\gamma}}}\nonumber\\
&=\uf[A^p\ho^{T-p}]\nonumber
\end{align*}}

\paragraph*{The NO case} We now proceed to prove that when the real answer is 
$\no$, our algorithm says $\no$ with a probability at least $1-\delta$. To 
prove this result (formally stated in 
Lemma~\ref{lem:no_instance}), we use as a tool the fact that when the integer 
program in~\eqref{online:lp:offline_ra} is in the $\no$ case where even a slack 
of 
$3\epsilon(1+\epsilon)$ will not make it feasible, then even the LP relaxation 
of~\eqref{online:lp:offline_ra} will be infeasible with a slack of $2\epsilon$. 
We prove this statement now by proving its contrapositive in 
Lemma~\ref{lem:no_case_ip_lp}.
\begin{lemma}\label{lem:no_case_ip_lp}
If the LP relaxation of~\eqref{online:lp:offline_ra} is feasible with a slack 
of $s$, then the integer program in~\eqref{online:lp:offline_ra} is 
feasible with a slack of $s(1+\epsilon) + \epsilon$. 
\end{lemma}
\begin{proof}
To prove this, we write the LP relaxation of the integer program 
in~\eqref{online:lp:offline_ra} 
slightly differently below. 
\begin{table}[!h]
\begin{center}
%\tbl{}{
\begin{align}
\label{lp:no_case_ip_lp_dual}
\text{Primal and dual LPs corresponding to integer program 
in~\eqref{online:lp:offline_ra}}
\end{align}
\begin{tabular}{l|l}
\textbf{Primal LP corresponding to IP~\eqref{online:lp:offline_ra}} & 
\textbf{Dual 
LP corresponding to IP~\eqref{online:lp:offline_ra}}\\\\
Minimize $\lambda\qquad$ s.t.  & Maximize $\sum_i (\rho_i - \alpha_i)
- \sum_{j} \beta_j\qquad$ s.t.\\\\
$\forall ~i,  \lambda - \sum_{j,k} \frac{\aijk\xjk}{c_i} \geq -
1$ & $\forall ~{j,k},\ \beta_j \geq \sum_{i} \left(
\rho_i\frac{\wijk}{d_i} - \alpha_i\frac{\aijk}{c_i}\right)$\\\\
$\forall ~i,  \lambda + \sum_{j,k} \frac{\wijk\xjk}{d_i} \geq  
1$ & $\sum_i (\alpha_i + \rho_i) \leq 1$\\\\
$\forall ~j, \sum_k \xjk \leq1$ &  $\forall\ i, \alpha_i, \dmi \geq 0$\\\\
$\forall ~j,k, \xjk \geq 0$ & $\forall\ j, \beta_j \geq 0$\\\\
$\lambda \geq 0$ & $\ $
\end{tabular}
%}
\end{center}
\renewcommand{\tablename}{LP}
\end{table}

The optimal value $\lambda^*$ of the primal LP in~\eqref{lp:no_case_ip_lp_dual} 
represents 
the 
slack in the YES/NO problem. I.e., if $\lambda^* = 0$, then we 
have zero slack and hence are in 
the YES case. Else, we are in the NO case with a slack equal to $\lambda^*$. 
Given this, all we have to show is that when the LP 
in~\eqref{lp:no_case_ip_lp_dual} has an optimal value of $\lambda^*$, then the 
corresponding integer program's optimal solution is at most 
$\lambda^*(1+\epsilon) + \epsilon$. To see this is true, let $\xjk^*$ 
denote the optimal solution to primal LP~\eqref{lp:no_case_ip_lp_dual}. 
Consider the 
integral solution that does a randomized rounding of the $\xjk^*$'s and 
allocates according to these rounded integers, and let $X_{jk}$ be the 
corresponding $\{0,1\}$ random variable. Let random variable $X_{ij} = \sum_k 
\frac{\aijk 
X_{jk}}{c_i}$. By the definition of LP~\eqref{lp:no_case_ip_lp_dual}, we have 
$\ex[\sum_j X_{ij}] \leq 1+\lambda^*$. Noting that each of the $X_{ij}$'s is at 
most $\gamma$, by 
Chernoff bounds it follows that $\prob[\sum_j X_{ij} \geq 
(1+\lambda^*)(1+\epsilon)] \leq e^{-\frac{\epsilon^2}{4\gamma}}$, which given 
that $\gamma = O(\frac{\epsilon^2}{\log(n/\epsilon)})$, is at most 
$\frac{\epsilon}{2n}$ (the probability 
derivation is just like the derivation 
in Section~\ref{online:sec:HOGWL}). Likewise, if we define by $Y_{ij}$ the 
random variable $\sum_k \frac{\wijk 
X_{jk}}{d_i}$, then by the definition of LP~\eqref{lp:no_case_ip_lp_dual}, we 
have 
$\ex[\sum_j Y_{ij}] \geq 1-\lambda^*$. By Chernoff bounds, we get an identical 
argument, we get $\prob[\sum_j Y_{ij} \leq 
(1-\lambda^*)(1-\epsilon)] \leq \frac{\epsilon}{2n}$. By doing a union bound 
over the $2n$ Chernoff bounds, we have that the randomized rounding integer 
solution is feasible with an optimal value of at most $\lambda^* + \epsilon + 
\lambda^*\epsilon$ with probability at least $1-\epsilon$. This means that 
there exists an integer solution of value $\lambda^* + \epsilon + 
\lambda^*\epsilon$, and this proves the lemma. 
\end{proof}
\begin{corollary}
\label{cor:mixed_pc}
For a $\no$ instance with a slack of $3\epsilon(1+\epsilon)$, the LP relaxation 
of the instance is still infeasible with a slack of $2\epsilon$. In particular, 
this implies that the optimal value of the primal $\lambda^*$ 
in~\eqref{lp:no_case_ip_lp_dual} is at least $2\epsilon$, and likewise the 
optimal dual value in~\eqref{lp:no_case_ip_lp_dual} is $\sum_i (\rho_i^* - 
\alpha_i^*)
- \sum_{j} \beta_j^* \geq 2\epsilon$.
\end{corollary}
\begin{lemma}\label{lem:no_instance}
For a $\no$ instance, if $T \geq \Theta(\frac{\gamma m\log(n/\delta)}{\epsilon^2})$, then 
 \[ \prob \left[ \max_i \frac{\sxia[T]}{c_i} < \frac{T}{m} (1 + \epsilon) 
\ \& \  \min_i \frac{\syia[T]}{d_i} > \frac{T}{m} (1 - \epsilon)\right] \leq \delta. \]
\end{lemma}
%The proof appears in the full version of the paper. 
%The proof is in appendix~\ref{sec:deferred_proofs}
\begin{proof}
Let $R$ denote the set of requests sampled. Consider the following LP.
\begin{table}[!h]
\begin{center}
%\tbl{}{
\begin{align}
\label{lp:sampled_dual}
\text{Sampled primal and dual LPs}
\end{align}
\begin{tabular}{l|l}
\textbf{Sampled primal LP} & \textbf{Sampled dual LP}\\\\
Minimize $\lambda\qquad$ s.t.  & Maximize $\frac{T}{m}\sum_i (\rho_i - \alpha_i)
- \sum_{j\in R} \beta_j\qquad$ s.t.\\\\
$\forall ~i,  \lambda - \sum_{j\in R,k} \frac{\aijk\xjk}{c_i} \geq -
\frac{T}{m}$ & $\forall ~{j \in R,k},\ \beta_j \geq \sum_{i} \left(
\rho_i\frac{\wijk}{d_i} - \alpha_i\frac{\aijk}{c_i}\right)$\\\\
$\forall ~i,  \lambda + \sum_{j\in R,k} \frac{\wijk\xjk}{d_i} \geq  \frac{T}{m}$ & $\sum_i (\alpha_i + \rho_i) \leq 1$\\\\
$\forall ~j\in R, \sum_k \xjk \leq1$ &  $\forall\ i, \alpha_i, \dmi \geq 0$\\\\
$\forall ~j,k, \xjk \geq 0$ & $\forall\ j\in R, \beta_j \geq 0$\\\\
$\lambda \geq 0$ & $\ $
\end{tabular}
%}
\end{center}
\renewcommand{\tablename}{LP}
\end{table}

%The question of whether or not our algorithm answers $\no$ is equivalent to whether or not the 
If the primal in LP~\eqref{lp:sampled_dual} has an optimal objective value at 
least $\frac{T\epsilon}{m}$, then by definition of our 
algorithm~\ref{online:alg:offline_ra}, we would have declared $\no$, i.e., if 
the sampled LP itself had a slack of $\epsilon$ (scaled by $\frac{T}{m}$), then 
no integral allocation based on those samples can obtain a smaller slack. We 
now 
show that
by picking $T = \Theta(\frac{\gamma m \ln(n/\delta)}{\epsilon^2})$, the above 
LP~\eqref{lp:sampled_dual} will have its optimal objective 
value at least $\frac{T\epsilon}{m}$, with a probability at least $1-\delta$. 
This makes our algorithm answer $\no$ with a probability at least $1-\delta$. 

Now, the primal of LP~\eqref{lp:sampled_dual} has an optimal value equal to that of the dual
%\begin{eqnarray}\label{lp:sampled_dual}
%\end{eqnarray}
which in turn is lower bounded by the value of dual
at any feasible solution. One such feasible solution is $\alpha^*, \beta^*, \rho^*$, which is the 
optimal solution to the full version of the dual in LP~\eqref{lp:sampled_dual}, 
namely the one written in LP~\eqref{lp:no_case_ip_lp_dual} where 
$R = [m]$, $T=m$. This is because the set of constraints in the full version of the dual is clearly a superset
of the constraints in the dual of LP~\eqref{lp:sampled_dual}.   
Thus, the optimal value of the primal of LP~\eqref{lp:sampled_dual} 
is lower bounded by value of dual
at $\alpha^*, \beta^*, \rho^*$, which is 
\begin{eqnarray}\label{eqn:no_case}
&=&\frac{T}{m}(\sum_i \rho_i^* - \alpha_i^*) - \sum_{j\in R}\beta^*_j
\end{eqnarray}
For proceeding further in lower bounding~\eqref{eqn:no_case}, we apply Chernoff
bounds to $\sum_{j\in R}\beta^*_j$. The fact that the dual of the full version 
of
LP~\eqref{lp:sampled_dual} is a maximization LP, coupled with the constraints
there in imply that $\beta_j^* \leq \gamma$. 
Further, let $\tau^*$ denote the optimal value of the full version of LP~\eqref{lp:sampled_dual}, i.e.,
$\sum_i (\rho_i^* - \alpha_i^*) - \sum_j \beta_j^* = \tau^*$. 
Now, the constraint  $\sum_i (\alpha_i^* + \rho_i^*) \leq 1$ 
%it follows that $\sum_i(\rho_i^* - \alpha_i^*) \geq -1$ and this
coupled with the fact that $\tau^* \geq 0$ implies $\sum_j \beta_j^* \leq 1$. 
%$\sum_j \beta_j^* \leq \tau^* + 1 \leq 2\max(\tau^*,1)$. 
We are now ready to lower bound the quantity in~\eqref{eqn:no_case}. We have the optimal solution to primal of LP~\eqref{lp:sampled_dual}  
\begin{eqnarray}\label{eqn:noCaseProof}
&\geq&\frac{T}{m}(\sum_i \rho_i^* - \alpha_i^*) - \sum_{j\in R}\beta^*_j\nonumber\\
&\geq& \frac{T}{m}\sum_i (\rho_i^* - \alpha^*_i) - \left(\frac{T\sum_j
\beta^*_j}{m} + \sqrt{\frac{4T(\sum_j\beta^*_j)\gamma \ln(1/\delta)}{m}}\right)
\text{ \bigg(Since $\beta_j^* \in [0,\gamma]$\bigg)}\nonumber\\
&\geq& \frac{T\tau^*}{m} -  \sqrt{\frac{4T\gamma \ln(1/\delta)}{m}}\nonumber\\
&=& \frac{T\tau^*}{m}\left[1-\sqrt{\frac{\gamma m \ln(1/\delta)}{T}\cdot\frac{4}{{\tau^*}^2}}\right]
\end{eqnarray}
where the second inequality is a ``with probability at least $1-\delta$'' inequality, i.e., we apply Chernoff bounds for $\sum_{j \in S}\beta^*_j$, along
with the observation that each $\beta^*_j \in [0, \gamma]$. The third inequality
follows from $\sum_j \beta_j^* \leq 1$ and $\sum_i (\rho_i^* - \alpha_i^*) -
\sum_j \beta_j^* = \tau^*$.  
Setting $T = \Theta(\frac{\gamma m \ln(n/\delta)}{\epsilon^2})$ with a appropriate constant inside the $\Theta$, coupled
with the fact that $\tau^* \geq 2\epsilon$ in the NO case (see 
Corollary~\ref{cor:mixed_pc}), it is easy to verify
that the quantity in~\eqref{eqn:noCaseProof} is at least $\frac{T\epsilon}{m}$.
%, both for $\tau^* \geq 1$ and for $\tau^* < 1$. This proves the lemma.

Going back to our application of Chernoff bounds above, in order to apply it in
the form above, we require that the multiplicative deviation from mean
$\sqrt{\frac{4\gamma m\ln(1/\delta)}{T\sum_j \beta_j^*}} \in [0,2e-1]$. If
$\sum_j \beta_j^* \geq \epsilon^2$, then this requirement would
follow. Suppose on the other hand that $\sum_j \beta_j^* <
\epsilon^2$.  Since we are happy if the excess over mean is
at most $\frac{T\epsilon}{m}$, let us look for a multiplicative error of
$\frac{\frac{T\epsilon}{m}}{\frac{T\sum_j\beta_j^*}{m}}$. Based on the fact that
$\sum_j \beta_j^* < \epsilon^2$ the multiplicative error can be seen to be
at least $1/\epsilon$ which is larger than $2e-1$ when $\epsilon <
\frac{1}{2e-1}$. We now use the version of Chernoff bounds for multiplicative error
larger than $2e-1$, which gives us that a deviation of $\frac{T\epsilon}{m}$
occurs with a probability at most
$2^{-\left(1+\frac{\frac{T\epsilon}{m}}{\frac{T\sum_j\beta_j^*}{m}}\right)\frac{T\sum_j\beta_j^*}{m\gamma}}$,
where the division by $\gamma$ is because of the fact that $\beta_j^* \leq
\gamma$. This probability is at most
$(\frac{\delta}{n})^{1/\epsilon}$, which is at most $\delta$.

\end{proof}

The proofs for the YES and NO cases together prove 
Theorem~\ref{online:thm:coveringPacking}.

\section{Special Cases of the Resource Allocation
Framework}\label{online:sec:special_cases}
%\label{sec:special_cases}
We now list the problems that are special cases of the resource allocation framework
and have been previously considered. The Adwords and Display ads special cases were already discussed in Section~\ref{online:subsec:online_ra}. 
%See the full version of the paper for more special cases. 

\subsection{Network Routing and Load Balancing}
Consider a graph (either undirected or directed) with edge capacities.  Requests
arrive online; a request $j$ consists of a source-sink pair, $(s_j,t_j)$ and a
bandwidth $\rho_j$.  In order to satisfy a request, a capacity of $\rho_j$ must
be allocated to it on every edge along some path from $s_j$ to $t_j$ in the
graph.  In the {\em throughput maximization} version, the objective is to
maximize the number of satisfied requests while not allocating more bandwidth
than the available capacity for each edge (Different requests could have
different values on them, and one could also consider maximizing the total value
of the satisfied requests). Our Algorithm~\ref{online:alg:online_ra} for
resource allocation framework directly applies here and the approximation
guarantee there directly carries over.  
~\citet{KPP} consider a different version of this problem where requests 
according to a Poisson process with unknown arrival rates. Each request has an 
associated \emph{holding time} that is assumed to be exponentially distributed, 
and once a request has been served, the bandwidth it uses up gets freed after 
its holding time (this is unlike our setting where once a certain amount of 
resource capacity has been consumed, it remains unavailable to all future 
requests). Some aspects of the distribution are assumed to be known, namely, 
that the algorithm knows the average rate of profit generated by all the 
incoming circuits, the average holding time, and also the target offline 
optimal offline solution that the online algorithm is aiming to approximate 
(again this is 
unlike our setting where no aspect of the request distribution is known to the 
algorithm, and there could even be some adversarial aspects like in the ASI 
model). When each request consumes at most $\gamma$ fraction of any edge's 
bandwidth,~\citet{KPP} give an online algorithm that achieves an expected 
profit of 
$(1-\epsilon)$ times the optimal offline solution when $\gamma = 
O(\frac{\epsilon^2}{\log n})$.

\subsection{Combinatorial Auctions} 
Suppose we have $n$ items for sale, with $c_i$ copies of item $i$.  Bidders
arrive online, and bidder $j$ has a utility function $U_j : 2^{[n]} \rightarrow
\Real$.  If we posted prices $p_i$ for each item $i$, then bidder $j$ buys a
bundle $S$ that maximizes $U_j(S) - \sum_{i\in S} p_i$. We assume that bidders
can compute such a bundle.  The goal is to maximize social welfare, the total
utility of all the bidders, subject to the supply constraint that there are
only $c_i$ copies of item $i$. Firstly, incentive constraints aside, this
problem can be written as an LP in the resource allocation framework.  The
items are the resources and agents arriving online are the requests. All the
different subsets of items form the set of options. The utility $U_j(S)$
represents the profit $w_{j,S}$ of serving agent $j$ through option $S$, i.e.
subset $S$. If an item $i \in S$, then $a_{i,j,S} = 1$ for all $j$ and zero
otherwise. Incentive constraints aside, our algorithm for resource allocation
at step $s$, will choose the option $k^*$ (or equivalently the bundle $S$) as
specified in point 8 of Algorithm~\ref{online:alg:online_ra}, i.e., minimize the
potential function. That is, if step $s$ falls in stage $r$,
\begin{equation*}
k^* = \arg \min_k\left\{\sum_i\aijk\cdot\phixir[s-1] - \wjk\cdot\psi_{s-1}^r\right\}
\end{equation*} 
(note that unlike Algorithm~\ref{online:alg:online_ra} there is no subscripting for $\wjk$). This can be equivalently written as
\begin{equation*}
k^* = \arg \max_k\left\{\wjk\cdot\psi_{s-1}^r - \sum_i\aijk\cdot\phixir[s-1]\right\}
\end{equation*} 
Now, maximizing the above expression at step $s$ is the same as picking the $k$ to maximize
$\wjk - \sum_i p_i(s)\aijk$, where %\[p_i =\frac{\frac{\epsc(\stage)/\gamma}{\left(1 + \frac{\epsc(\stage)}{\gamma m}\right)}}{\frac{\epso(\stage)/\wm}{\left(1 - \frac{\epso(\stage) Z(\stage)}{\wm m}\right)} \phi_{\text{obj}}^{t-1}}\cdot\frac{\phi_i^{t-1}}{c_i}. \]
$p_i(s) = \frac{\phixir[s-1]}{\psi_{s-1}^r}$.
Thus, if we post a price of $p_i(s)$ on item $i$ for bidder number $s$, he will do exactly what the algorithm would have done otherwise. 
Suppose that the bidders are i.i.d samples from some distribution (or they arrive as in the \asi\ model). 
We can use Theorem~\ref{online:thm:online_ra} 
to get an incentive compatible posted price auction\footnote{Here we assume that each agent reveals his true utility function \emph{after} he makes his purchase. This information is necessary to compute the prices to be charged for future agents.} with a competitive ratio of $1-O(\epsilon)$ 
whenever $\gamma = \min_i\{c_i\} \geq \Omega\left(\frac{\log(n/\epsilon)}{\epsilon^2}\right).$
Further if an analog of Theorem \ref{online:thm:online_ra} also holds in the random permutation model 
then we get a similar result for combinatorial auctions in the offline case: 
we simply consider the bidders one by one in a random order.

\bibliographystyle{plainnat}
\bibliography{dissertation}

\end{document}